\documentclass[11pt]{article}

\usepackage{fullpage}

\usepackage{cite}
\usepackage{graphicx}
\usepackage{amsmath}
\usepackage{array}
\usepackage{multirow}
\usepackage{xcolor}

\usepackage{balance}

\usepackage{latexsym}
\usepackage{url}
\usepackage{xspace}

\usepackage{tikz}

\newcommand{\eg}{e.g.\@\xspace}

\newtheorem{definition}{Definition} 
 
\newtheorem{lemma}{Lemma}
\newtheorem{proposition}{Proposition} 
\newtheorem{theorem}{Theorem}
\newtheorem{corollary}{Corollary}
\newenvironment{proof}{\emph{Proof:}}{$\Box$\newline}

\newcommand{\ceiling}[1]{\lceil #1\rceil}
\newcommand{\floor}[1]{\lfloor #1\rfloor}

\newcommand{\cost}{\mathrm{Time}} 
 
\newcommand{\size}{\mathrm{Size}}
\newcommand{\comm}{\mathrm{comm}} 
\newcommand{\theroot}{\mathrm{root}}

\newcommand{\mpigather}{\texttt{MPI\_\-Gather}\xspace}
\newcommand{\mpigatherv}{\texttt{MPI\_\-Gatherv}\xspace}
\newcommand{\mpiscatter}{\texttt{MPI\_\-Scatter}\xspace}
\newcommand{\mpiscatterv}{\texttt{MPI\_\-Scatterv}\xspace}
\newcommand{\tuwgatherv}{\texttt{TUW\_\-Gatherv}\xspace}
\newcommand{\tuwscatterv}{\texttt{TUW\_\-Scatterv}\xspace}
\newcommand{\twogatherv}{\texttt{TWO\_\-Gatherv}\xspace}
\newcommand{\twoscatterv}{\texttt{TWO\_\-Scatterv}\xspace}
\newcommand{\mpiint}{\texttt{MPI\_\-INT}\xspace}

\title{On Optimal Trees for Irregular Gather and Scatter Collectives}
 
\author{Jesper Larsson Tr\"{a}ff\\
  TU Wien (Vienna University of Technology), Faculty of Informatics \\
  Vienna, Austria,\url{traff@par.tuwien.ac.at}
  }

\begin{document}
\maketitle

\begin{abstract}
We study the complexity of finding communication trees with the lowest
possible completion time for rooted, irregular gather and scatter
collective communication operations in fully connected, $k$-ported
communication networks under a linear-time transmission cost model.
Consecutively numbered processors specify data blocks of possibly
different sizes to be collected at or distributed from some (given)
root processor where they are stored in processor order.  Data blocks
can be combined into larger segments consisting of blocks from or to
different processors, but individual blocks cannot be split.  We
distinguish between ordered and non-ordered communication trees
depending on whether segments of blocks are maintained in processor
order.  We show that lowest completion time, ordered communication
trees under one-ported communication can be found in polynomial time
by giving simple, but costly dynamic programming algorithms.  In
contrast, we show that it is an NP-complete problem to construct
cost-optimal, non-ordered communication trees. We have implemented the
dynamic programming algorithms for homogeneous networks to evaluate
the quality of different types of communication trees, in particular
to analyze a recent, distributed, problem-adaptive tree construction
algorithm.  Model experiments show that this algorithm is close to the
optimum for a selection of block size distributions.  A concrete
implementation for specially structured problems shows that optimal,
non-binomial trees can possibly have even further practical advantage.
\end{abstract}

\paragraph{Keywords:}
Gather and Scatter Collective Operations, Irregular Collective
Operations, Communication trees, $k$-ported communication networks,
Dynamic programming, NP-completeness.

\maketitle
\section{Introduction}
Collective gather and scatter operations play a role in many parallel
applications for distributing or collecting data between a designated
(root) process and other processes in the application.  Gather and
scatter operations are therefore included as collective operations in
most interfaces and languages for parallel and distributed computing,
notably in MPI~\cite{MPI-3.1}, but for also in PGAS languages
and frameworks like
UPC++~\cite{ElGhazawiCarlsonSterlingYelick05,ZhengKamilDriscollShanYelick14},
and Global Arrays~\cite{Nieplocha06}.  In these gather and scatter
operations, a single process (thread) in a set of processes (threads)
is to either collect data blocks from or distribute data blocks to all
(other) processes (threads). While good algorithms for the regular
variants of the problems where all data blocks have the same size are
known for many types of communication networks under different
communication
models~\cite{AlexandrovIonescuSchauserScheiman97,ChanHeimlichPurkayasthavandeGeijn07,JohnssonHo89},
this is much less the case for the irregular variants where different
processes may contribute blocks of different sizes.

Collective communication operations, in particular irregular
operations where processes contribute different amounts of data, pose
two different algorithmic problems. The first is to determine for
a(ny) given input data block distribution to the operation how fast
the operation can be carried out, preferably by a closed-form
expression. The second is to determine the complexity of building a
communication schedule and structure that will allow to solve the
problem in the determined time. For many regular operations, the
latter problem can be solved in constant time per process, or with
only a small, acceptable, say, logarithmic overhead, while yielding
optimal completion time solutions. This is most often not the case for
irregular operations.

For gather and scatter operations, trees are natural communication
structures since data blocks flow to or from a single root process.
This paper contributes to clarify the complexity of finding optimal
(fastest, lowest completion time) communication trees for irregular
gather and scatter problems under specific communication network
assumptions that may serve as useful enough first approximations to
real interconnects and communication systems. In particular we show
that optimal trees can be constructed in polynomial time under certain
natural constraints on how trees are built, whereas without these
constraints, optimal tree construction is an NP-hard problem.

We study the problems under a linear transmission cost model, where
the cost of transmitting a data block from one processor to another is
proportional to the size of the block plus some constant start-up
latency. Processors can only be involved in a single, or a small
number of communication operations at a time. We assume that any
processor can communicate with any other processor, but the cost of
communication may be different for different pairs. Processes are
bound one-to-one to processors, and ranked consecutively from $0$ to
$p-1$, $p$ being the number of processors in the network. That is, our
communication model is the fully connected, $k$-ported,
non-homogeneous, linear cost communication
network~\cite{FraigniaudLazard94}.

We distinguish between two types of communication trees. In the gather
operation, data blocks from all processors are collected as a
consecutive segment at the root processor in \emph{rank
  order}. Conversely, for the scatter operation blocks stored in rank
order at the root processor are distributed to the other processors
such that processor $i$ eventually receives the $i$th block from the
root. An \emph{ordered} gather or scatter tree has the property that
segments of data blocks at processors that are interior tree nodes are
consecutive and in rank order, that is blocks for some processors
$j,j+1,\ldots,j+s$ for $s\geq 0$ with $i\in [j,\ldots,j+s]$ for
processor $i$.  This has the advantage that processors that have to
send blocks further on will never have to perform possibly costly,
local reorderings of these blocks, and can make the implementation
for, say, MPI, easier. However, this restriction may exclude
communication schedules with lower completion times.  For
\emph{non-ordered} trees, this constraint is dropped, and processors
are allowed to send or receive not necessarily rank ordered segments
of blocks in any order. Non-ordered trees may require reordering of
blocks into rank order, at least at the root processor.

Concrete contributions of the paper are the following:
\begin{itemize}
\item
For any given gather or scatter input instance, we show that with
one-ported communication, optimal (fastest, lowest completion time),
ordered communication trees can be found in polynomial time by a
simple, dynamic programming algorithm running in $O(p^3)$ operations
for the homogeneous communication cost case, and in $O(p^4)$
operations for the general, non-homogeneous case.
\item
We also show that optimal, binary communication trees can be computed
by the same dynamic programming approach; all algorithms extend easily
to the simpler, related problems of broadcast and reduction of data
blocks (vectors).
\item
We use the offline, dynamic programming constructions to compare the
completion times for different types of gather and scatter trees,
in particular
showing that a recently proposed, simple $\ceiling{\log p}$
communication round, bottom up algorithm~\cite{Traff18:irreggatscat}
can achieve good results compared to the optimal solutions.
\item
For specially structured problems consisting of two sizes of data
blocks that are either large or small, we indicate that optimal,
(non-)ordered trees can indeed perform (much) better on a concrete system
than the adaptive binomial trees generated by the algorithm
in~\cite{Traff18:irreggatscat}.
\item
Finally, we show that computing optimal, non-ordered communication
trees is an NP-complete problem, meaning that the problems of finding
optimal ordered and non-ordered trees are computationally
different. 
\end{itemize}

\subsection{Related work}

Results on the gather and scatter collective communication operations
can be found scattered over the literature. Standard, binomial tree
and linear algorithms for the regular problems that are indeed used in
most MPI library implementations, assuming a homogeneous, linear
transmission cost model, or a simple hierarchical system are described
in,
\eg,~\cite{ChanHeimlichPurkayasthavandeGeijn07,KielmannBalGorlatchVerstoepHofman01}. Extensions
to multiple communication ports for some of these algorithms have been
proposed in, \eg,~\cite{ChanvandeGeijnGroppThakur06,SackGropp15}.
Algorithms for the regular scatter operation for the $LogP$ and
$LogGP$ models were discussed
in~\cite{AlexandrovIonescuSchauserScheiman97}. Other algorithms and
implementations for the regular problems for MPI and UPC for a
specific processor architecture can be found
in~\cite{MallonTaboadaKoesterke16}.  Approaches to the irregular
problems for MPI can be found in,
\eg,~\cite{Traff04:gatscat,DichevRychovLastovetsky10,Traff18:irreggatscat},
but it is fair to say that there has been overall little attention
paid to these problems in the MPI community.  More theoretical papers
consider the problems in a different setting, \eg, that of finding
optimal communication schedules for given
trees~\cite{BhattPucciRanadeRosenberg93}.

\section{The Model and the Problems}

For now we leave concrete systems and interfaces like MPI aside but
return to specific issues in later remarks. We define the
communication models and problems in terms of processors carrying out
communication operations.
Since the gather and scatter operations are
semantically ``dual'', we mostly treat only one of them;
the results translate into analogous results for the other.

\subsection{Communication network model}

We assume a \emph{fully connected network} of $p$ communication
processors ranked consecutively from $0$ to
$p-1$~\cite{FraigniaudLazard94}. Processors can communicate pairwise
in a synchronous, point-to-point fashion with one processor sending to
another, receiving processor, both being involved during the
transmission.  Communication costs are \emph{linear} but not
necessarily homogeneous.  More precisely, the time for transmitting a
message of $m$ units (Bytes) from processor $i$ to processor $j$ is
modeled as $\alpha_{ij}+\beta_{ij}m$ where $\alpha_{ij}$ is a
\emph{start-up latency} for the communication and $\beta_{ij}$ a
\emph{time per transmitted unit}. Furthermore, each processor will
have a \emph{local copy cost} of $\gamma_i m$ for copying a data block
of $m$ units from one buffer to another.  Communication is
\emph{$k$-ported}, $k\geq 1$, meaning that a communication processor
can be engaged in at most $k$ communication operations at a time
during which it is fully occupied. All $\floor{p/2}$ processor pairs
can communicate at the same time. The cost of a communication
algorithm is the time for the slowest processor to complete
communication and local copying under the assumption that all
processors start at the same time.

\subsubsection*{Discussion}

The communication network model captures certain features of modern
high-performance systems, neglects others, and perhaps misrepresents
some. It is useful, if it leads to algorithms that perform well on
real systems (compared to other algorithms, designed under other
assumptions), and if it serves to clarify inherent complexities in the
gather and scatter problems.

The assumption that point-to-point communication is synchronous with
both processors involved when communication takes place is common
(albeit sometimes
implicit)~\cite{Bruck97,ChanHeimlichPurkayasthavandeGeijn07}, and
leads to contention free algorithms for truly fully connected,
one-ported networks.  Since it fixes 
when communication takes place, it can make the analysis of collective
algorithms significantly easier than in possibly more realistic models
like $LogGP$~\cite{AlexandrovIonescuSchauserScheiman97} that allow
outstanding communication operations and overlap of both communication
and computation. It often allows the development of optimal
algorithms which is also significantly more difficult under $LogGP$,
where few optimality results are known.  Algorithms designed under the
synchronous assumption often perform well in practice, where sometimes
strict synchronous communication is relaxed for better
performance. Models like $LogP$ and $LogGP$ can lead to algorithms
with unrealistically large numbers of outstanding communication
operations, unless some capacity constraint is externally imposed, and
contention has to be accounted for. Ironically, the optimal scatter
algorithm in the simple $LogP$ model has the root send data blocks to
the other processors one after the
other~\cite{AlexandrovIonescuSchauserScheiman97}. This does not
correspond to practical
experience~\cite{ThakurGroppRabenseifner05,Traff18:irreggatscat}, and
was one effect motivating the $LogGP$ model.

The non-homogeneous assumption makes it possible to model some aspects
of systems where routing between some processors $i$ and $j$ is
needed, or of clustered, hierarchical systems with different
communication characteristics inside and between compute nodes.
Sparse, non-fully connected networks can be captured by setting
latency $\alpha_{ij}=\infty$ for processors $i$ and $j$ that are not
connected in the network. The model cannot account for congestion in
such networks, though. In such sparse networks, there may not be
feasible solutions to the ordered gather and scatter problems. In
order to guarantee feasible solutions in all cases, instead
$\alpha_{ij}$ and $\beta_{ij}$ can be chosen to reflect the time for
routing from processor $i$ to processor $j$. The model could be
extended to piecewise linear transmission costs with different
$\alpha_{ij}$ and $\beta_{ij}$ values for different message ranges.

With $k>1$ communication ports in the model adopted here, processors
can be involved in up to $k$ concurrent communication operations at a
time. Most of our results in the following will be for $k=1$. The
general case with $k\geq 1$ is more difficult for reasons that will be
pointed out. A number of results for regular collective communication
operations in $k$-ported (torus) systems can be found
in,\eg,~\cite{Bruck97,ChanvandeGeijnGroppThakur06,SackGropp15}.

The communication system is \emph{homogeneous} if the communication
costs for all processor pairs are characterized by the same start-up
latency $\alpha$ and time per unit $\beta$. Likewise, homogeneous
processors will have the same local copy cost $\gamma$.  In our model,
point-to-point communication is always out of or into consecutive
communication buffers, which can necessitate local copies or
reorderings of data blocks. In interfaces like MPI, this can sometimes
be done implicitly by the use of derived datatypes~\cite[Chapter
  4]{MPI-3.1}.  It is probably realistic to assume that a local copy
can at least partially be performed concurrently with communication,
and our results can be adopted to this. However, for simplicity we
assume here that there is such a cost governed by $\gamma\geq 0$ that
cannot be overlapped with communication\footnote{Taking $\gamma=0$ is
  not strictly the same as assuming that a cost with $\gamma>0$ can be
  overlapped with communication; if communication is fast, some part
  of the local copy cost $\gamma m$ cannot be overlapped and will have
  to be paid.}. We tacitly assume that $\gamma_i\leq\beta_{ij}$ for
any processors $i$ and $j$ in order to prevent artificial algorithms
where some local data are sent back and forth between processors in
order to save on local copy costs.

Albeit pairwise communication is synchronous, the overall execution of
an algorithm is not. Each processor and communication port can engage
in new communication as soon as it has completed its previous
communication operation, independently of what the other processors
are doing. This means that delays can be incurred when the
communication partner is not yet ready.  Optimal algorithms will
minimize the overall effects of such delays.  Note that under this
asynchronous model, the information dissemination lower bound argument
from round-based, synchronous models of $\ceiling{\log_2 p}$
communication rounds~\cite{Bruck97} does not apply. Optimal gather and
scatter trees may well have a root degree smaller than
$\ceiling{\log_2 p}$. Lemma~\ref{lem:largesmall} gives an example.

\subsection{The irregular gather and scatter operations}

The \emph{irregular gather/scatter operations} are the following.
Each processor has a local communication buffer of size $m_i$ units
(Bytes). In addition, a designated \emph{root processor} $r, 0\leq
r<p$ has a buffer of size $m=\sum_{i=0}^{p-1}m_i$ capable of storing
data for all processors, including the root itself. We will refer to
$m$ as the \emph{size} of the gather/scatter problem.  A gather or
scatter problem is said to be \emph{regular}, if all processors have
the same buffer size $m_i=m/p$.

In the \emph{gather problem}, each processor has a \emph{data block}
$[m_i]$ of size $m_i$ in its communication buffer, and the root has to
collect all data blocks consecutively in rank order into a large
segment of blocks
$[m_0,\-\ldots,\-m_{r-1},\-m_r,\-m_{r+1},\-\ldots,\-m_{p-1}]$. The
\emph{scatter problem} is the opposite: The root has a large,
consecutive segment of blocks in rank order
$[m_0,\-\ldots,\-m_{r-1},\-m_r,\-m_{r+1},\-\ldots,\-m_{p-1}]$ in its
buffer, and has to distribute the blocks to the processors such that
processor $i$ eventually has the data block $[m_i]$ in its local
buffer.

For both gather and scatter operations, the root processor $r$ has a
data block $[m_r]$ to itself.

\subsubsection*{Discussion}

In common interfaces like MPI and UPC, the root is always an
externally given, fixed processor. Alternatively, the root could be
decided by the algorithm solving the gather/scatter problem and lead
to faster completion time. We consider both variations here.
Concrete interfaces may give more control over the placement of data
blocks at the root process. In MPI, for instance, the actual placement
is controlled by an explicit offset for each block. The internal
structure of blocks can likewise be controlled via MPI derived
datatypes. Such features, however, do not change the essential
algorithmic costs of the operations.  The MPI gather/scatter
operations also assume that the root process has a local block to
itself which has to be copied from one buffer to another, unless the
\texttt{MPI\_IN\_PLACE} option is supplied~\cite[Section
  5.5]{MPI-3.1}).

\begin{figure*}
\begin{eqnarray}
  \cost(T^r_{\{r\}}) & = & 0 \nonumber \\
\cost(T^r_R(T^{r_0}_{R_0},T^{r_1}_{R_1},\ldots,T^{r_{j-1}}_{R_{j-1}},T^{r_j}_{R_j})) & = & 
\max(\cost(T^r_R(T^{r_0}_{R_0},T^{r_1}_{R_1},\ldots,T^{r_{j-1}}_{R_{j-1}})),\cost(T^{r_j}_{R_j}(\cdot))) + \nonumber\\
& & \alpha_{r_jr}+\beta_{r_jr}\size(T^{r_j}_{R_j}) \label{eq:commcost}\\
\cost(T^r_R(T^{r_0}_{R_0},T^{r_1}_{R_1},\ldots,T^{r_{j-1}}_{R_{j-1}},T^{r}_{\{r\}})) & = & 
\cost(T^r_R(T^{r_0}_{R_0},T^{r_1}_{R_1},\ldots,T^{r_{j-1}}_{R_{j-1}}))+\gamma_rm_r 
\label{eq:localcopy}
\end{eqnarray}
\caption{Equations defining the completion time of a gather or scatter
  tree $T^r_R$ over processors in $R$ rooted at processor $r\in R$
  with (prefix of) subtrees
  $(T^{r_0}_{R_0},T^{r_1}_{R_1},\ldots,T^{r_i}_{\{r_i\}})$ in that
  order. One subtree must be $T^r_{\{r\}}$ for which the local copy at
  the root $r$ is done. $T^{r_j}_{R_j}(\cdot)$ denotes the $j$th
  subtree with its full sequence of subtrees.}
\label{fig:treecost}
\end{figure*}

\subsection{Lower bounds}
\label{sec:lowerbounds}

An obvious lower bound for both scatter and gather operations with
root processor $r, 0\leq r<p$ in a one-ported, homogeneous
communication cost model is $\alpha+\beta(\sum_{i\neq r}m_i)+\gamma
m_r$ since all data blocks have to be sent from or received at the
root from some (or several) processor(s), except for the root's own
block for which a local copy cost has to be paid. With non-homogeneous
communication costs, a lower bound for, \eg, the scatter operation is
$\min_{j\neq r}(\alpha_{rj}+\beta_{rj}\sum_{i\neq r}m_i)+\gamma_r
m_r$, determined by the fastest reachable neighbor processor of the root.

With instead $k>1$ communication ports that can work simultaneously, a
lower bound for the homogeneous cost case becomes
$\alpha+\beta(\ceiling{\sum_{i\neq r}m_i/k})+\gamma m_r$, assuming that
data blocks $[m_i]$ can be arbitrarily split.

In the algorithms we consider here, this will be forbidden. Blocks can
be compounded into larger segments of blocks, but individual blocks
$[m_i]$ cannot be further subdivided. In this case, a lower bound is
less trivial to formulate. Let $P_0,P_1,\ldots P_{k-1}$ be a partition
of $\{0,\ldots p-1\}\setminus\{r\}$ into $k$ subsets. Assume that the
data blocks are distributed over $k$ processors such that processor
$i$ has the blocks in $P_i$. In that case, the gather and scatter
operations can be completed in time $\alpha+\max_{0\leq
  j<k}\beta\sum_{i\in M_j}m_i+\gamma m_r$. A lower bound is given by
a best such partition, \eg, by a partition that minimizes this
gather/scatter time.

\subsubsection*{Discussion}

For the $k$-ported lower bound, we assumed that the communication
processor can with only one start-up latency $\alpha$ initiate or
complete $k$ communication operations. This latency $\alpha$ may
depend on $k$. It could instead be assumed that each initiated
communication operation would occur its own start-up latency
$\alpha$. See~\cite{Bruck97,ChanvandeGeijnGroppThakur06,SackGropp15}
for further discussion.

The lower bound for $k$-ported communication indicates that attaining
it implicitly requires solving a possibly hard packing
problem. Section~\ref{sec:hardness} shows that this is the case even
in the one-ported case, unless the allowed algorithms are further
restricted.

\subsection{Gather and scatter trees}

The trivial algorithms for the gather and scatter operations let the
root processor receive or send the data blocks from or to the other
processors in some fixed order. The trivial algorithms have a high
latency term of $(p-1)\alpha$. Possibly better algorithms let
processors collect larger segments of data blocks that are later
transmitted either as a whole or as smaller segments.  We consider
such algorithms here, but restrict attention to trees in the following
strict sense.

\begin{definition}
  \label{def:gatherscattertrees}
  Let the $p$ processors be organized in a tree rooted at root
  processor $r$.  A \emph{gather tree algorithm} allows each processor
  except $r$ to perform a single send operation (of a segment of
  blocks, towards the root).  A \emph{scatter tree algorithm} allows
  each processor except $r$ to perform a single receive operation (of
  a segment of blocks).
\end{definition}

\subsubsection*{Discussion}

The restriction to trees means that no individual blocks $[m_i]$ are
allowed to be split, since this would imply that some processors in a
gather tree perform several send operations. In the one-ported,
homogeneous transmission cost model, this restriction is not serious,
since communication time between two processors cannot be improved by
pipelining, and pipelining a block through a longer path of processors
cannot be faster than sending the block directly to the root. It also
prevents algorithms that send some blocks several times, but this
would be redundant anyway and cannot improve the gather completion
time. All current implementations for the MPI gather and scatter
operations are such tree algorithms (to the authors knowledge).

It is possible that a gather operation could be performed faster,
especially with non-homogeneous communication costs by algorithms
where processors collect segments of blocks that are then split into
smaller segments and sent through different paths to the root (even
without violating the restriction that individual blocks are not
split). The author knows of no such DAG (Directed Acyclic Graph)
algorithms for the gather and scatter operations.

It is well-known that optimal trees for regular gather and scatter
problems in the linear transmission cost model are binomial trees with
root degree $\ceiling{\log_2 p}$ when $\alpha>0$, see,
\eg~\cite{ChanHeimlichPurkayasthavandeGeijn07}. As will be seen, this
is not the case for the irregular problems.

\subsection{Gather and scatter tree completion times}

We now formally define the completion time (cost) of gather and
scatter trees for the irregular operations. We (first) restrict
the communication system to be one-ported.

Let $T^r_R$ be a gather/scatter tree spanning a set of processors $R$,
rooted at processor $r\in R$. The \emph{size} of $T_R$ is defined to
be the sum of the sizes of the data blocks for the processors in the
tree, $\size(T^r_R) = \sum_{i\in R}m_i$. Unless $T^r_R$ is a singleton
tree $T^r_{\{r\}}$, the tree has some $j+1$ rooted subtrees as
children, $T^{r_0}_{R_0},T^{r_1}_{R_1},\ldots,T^{r_j}_{R_j}$, for
disjoint sets $R_0,R_1,\ldots,R_k$ with $r_i\in R_i$ that partition
$R$. By $T^r_R(T^{r_0}_{R_0},T^{r_1}_{R_1},\ldots,T^{r_j}_{R_j})$ we
denote the tree with (a prefix of) its subtrees in that order.
Since the root by the one-ported communication assumption must
communicate with these subtrees one after the other, the subtrees are
considered in sequence, one after the other, and this order is given
as part of the gather/scatter tree. Since local copying cannot by our
assumptions be overlapped with communication, the root $r$ has at some
point to perform its local copy of its data block $[m_i]$ which we
account for by having one of the subtrees be the singleton
$T^r_{\{r\}}$ for which it holds that $\size(T^r_{\{r\}})=m_r$.

\begin{definition}
\label{def:treecost}
The \emph{completion time} (cost),
$\cost(T^r_R(T^{r_0}_{R_0},\-T^{r_1}_{R_1},\-\ldots,\-T^{r_j}_{R_j}))$,
of a gather or scatter tree $T^r_R$ with a full sequence of subtrees
$T^{r_0}_{R_0},T^{r_1}_{R_1},\ldots,T^{r_j}_{R_j}$ in that order under
the one-ported network model is defined by the equations given in
Figure~\ref{fig:treecost}.  A tree $T^r_R$ is \emph{optimal} if it has
least completion time over all possible trees over $R$ with root $r$.
\end{definition}

The equations in Figure~\ref{fig:treecost} express that before
processor $r$ can gather from its $i$th subtree $T^{r_j}_{R_j}$, it
must have completed gathering from the previous subtrees. When also
the gathering in $T^{r_j}_{R_j}$ by processor $r_j$ has been
completed, the segment of data blocks from $T^{r_j}_{R_j}$ can be sent
from processor $r_j$ to processor $r$ at the cost given by the
(non-homogeneous) transmission cost model
(Equation~(\ref{eq:commcost})). If the $j$th subtree is $r$ itself,
the local copy has to be done (Equation~(\ref{eq:localcopy})).  For
scattering, the root first sends the segment of data blocks for the
$j$th subtree to processor $r_j$, and then scatter to the preceding
subtrees. The completion time is the transmission time plus the time
for the slower of the two concurrent scatter operations.

We note that given a gather or scatter tree as described, the
completion time of the operation can be computed in $O(p)$ time steps
by a bottom up traversal of the tree.

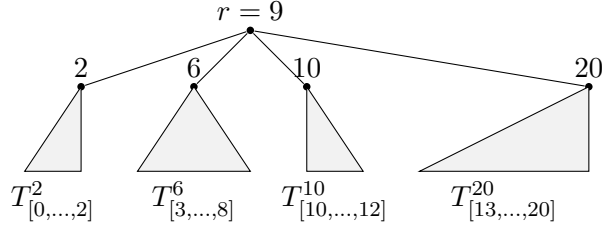
\begin{figure}
\begin{center}
\begin{tikzpicture}[scale=0.75]
\node (r) at (4,2.5) [fill,circle,inner sep=1pt] {}; \draw (r) node
      [above] {$r=9$};

\node (r0) at (1,1.5) [fill,circle,inner sep=1pt] {}; 
\node (i0) at (0,0) {}; 
\node (j0) at (1,0) {};

\draw[draw=none] (i0) -- (j0) node [midway,below] {$T^2_{[0,\ldots,2]}$};
\draw[fill=gray!10] (j0.center) -- (r0.center) -- (i0.center) -- cycle;
\draw (r0) node [above] {$2$};

\node (r1) at (3,1.5) [fill,circle,inner sep=1pt] {};
\node (i1) at (2,0) {};
\node (j1) at (4,0) {};

\draw[draw=none] (i1) -- (j1) node [midway,below] {$T^6_{[3,\ldots,8]}$};
\draw[fill=gray!10] (j1.center) -- (r1.center) -- (i1.center) -- cycle; 
\draw (r1) node [above] {$6$};

\node (r2) at (5,1.5) [fill,circle,inner sep=1pt] {}; 
\node (i2) at (5,0) {};
\node (j2) at (6,0) {};

\draw[draw=none] (i2) -- (j2) node [midway,below] {$T^{10}_{[10,\ldots,12]}$};
\draw[fill=gray!10] (j2.center) -- (r2.center) -- (i2.center) -- cycle; 
\draw (r2) node [above] {$10$};

\node (r3) at (10,1.5) [fill,circle,inner sep=1pt] {}; 
\node (i3) at (7,0) {}; 
\node (j3) at (10,0) {};

\draw[draw=none] (i3) -- (j3) node [midway,below] {$T^{20}_{[13,\ldots,20]}$};
\draw[fill=gray!10] (j3.center) -- (r3.center) -- (i3.center) -- cycle;
\draw (r3) node [above] {$20$};

\draw (r0) -- (r); 
\draw (r1) -- (r);
\draw (r2) -- (r); 
\draw (r3) -- (r);
\end{tikzpicture}
\end{center}
\caption{An ordered gather or scatter tree for $p=21$ processors
  rooted at processor $r=9$. The tree is neither binary nor binomial.
  In a gather operation, the root receives consecutive segments of
  data blocks from the subtrees from left to right. In a scatter
  operation, the root sends consecutive segments to the subtrees from
  right to left.  Whether the tree is optimal depends on the data
  block sizes $m_i$.}
\label{fig:orderedtree}
\end{figure}

In order to classify the complexity of constructing optimal gather and
scatter trees we now introduce a further constraint on gather and
scatter trees.

\begin{definition}
\label{def:ordering}
A gather or scatter communication tree $T^r_R$ is \emph{(strongly)
  ordered} if $R=[i,i+1,\ldots,j]$ is a list of consecutively ranked
processors from $i$ to $j, j\geq i$, $r\in R$, each subtree is
likewise strongly ordered, and the subtrees are sequenced one after
the other in such a way that the last processor in any subtree is the
processor ranked immediately before the first processor in the
immediately following subtree.

A gather or scatter tree $T^r_R$ that is not strongly ordered
\emph{non-ordered}.
\end{definition}

\subsubsection*{Discussion}

Ordered versus non-ordered are structural properties of the
communication trees. The completion time of
Definition~\ref{def:treecost} assumes $k=1$ communication ports, but
can be extended also to $k>1$ communication ports.  Each port can
independently receive or send blocks from a subset of the subtrees,
one after the other, and the completion time would be the time for the
last port to finish. Subtrees could be assigned to ports
statically. Alternatively, it could be assumed that the assignment is
done greedily, such that a finished ports starts sending or receiving
from the next subtree in the sequence not assigned to a port. Note
that the precise choice would not change the cost of an optimal
completion time tree.

An ordered communication tree with four (ordered) subtrees is shown in
Figure~\ref{fig:orderedtree}. In an ordered gather tree algorithm,
each processor gathers and maintains only a consecutively ordered
segment of data blocks from its children, and at no stage will blocks
have to be permuted to fulfill the consecutive ordering constraint at
the root.  Also no prior offset calculations are needed, the next
block segment can be placed in the communication buffer immediately
after the already received blocks.  If the ordered subtrees are not
sequenced one after another as defined, and the block segments are
received in some possibly non-consecutive order (as could be the case
in a non-blocking, asynchronous communication model), the tree is
still said to be \emph{weakly ordered}.  Note that a gather tree is
strongly ordered if and only if the strongly ordered subtrees are
sequenced such that the roots are in increasing rank order.

Non-ordered trees provide more freedom to reduce completion time, but
permutation of blocks at the root or other non-leaf processors to put
blocks into rank order could lead to extra costs of $\gamma_r m$ or
more. Still note that an optimal, non-ordered tree may complete faster
than an optimal, weakly ordered tree, which may in turn complete
faster than an optimal, ordered tree.

In~\cite{Traff18:irreggatscat} it was shown that ordered, binomial
gather and scatter trees for the homogeneous transmission cost model
can be found efficiently in a distributed manner with a root processor
chosen by the algorithm. We use this result here, and also compare
these trees against other (optimal) trees in our cost model in
Section~\ref{sec:quality}. These trees are good, but (usually) not
optimal; but (as will be seen) better trees require a large effort to
construct.
\begin{proposition}
\label{prop:linear}
Let $m_i,i=0,\ldots,p-1$ be the data block sizes for the $p$
processors. With homogeneous communication costs, and local copy cost
$\gamma=\beta$, strongly ordered gather and scatter trees with
completion time $\ceiling{\log_2 p}\alpha+\beta\sum_{i=0}^{p-1}m_i$
exist, and can be constructed offline in $O(p)$ steps.
\end{proposition}

The trick of the construction is to pair adjacent, ordered subtrees
with the same number of processors, but such that the subtree with the
smaller amount of data sends (in the gather case) its data to the tree
with the larger amount of data.

\section{Polynomial Time Constructions for Ordered Trees}

\begin{figure}
\begin{center}
\begin{tikzpicture}[scale=0.75]
\node (r) at (0.5,1.5) [fill,circle,inner sep=1pt] {};
\node (i) at (0,0) {}; 
\node (j) at (1.5,0) {};

\draw[draw=none] (i) -- (j) node [midway,below] {$T^r_R$}; 
\draw[fill=gray!10] (j.center) -- (r.center) -- (i.center) -- cycle; 
\draw (r) node [above] {$r\in R$};

\node (r1) at (3.5,1) [fill,circle,inner sep=1pt] {}; 
\node (i1) at (3,0) {};
\node (j1) at (4,0) {};

\draw[draw=none] (i1) -- (j1) node [midway,below] {$T^{r'}_{\bar{R}}$}; 
\draw[fill=gray!10] (j1.center) -- (r1.center) -- (i1.center) -- cycle; 
\draw (r1) node [anchor=south west] {$r'\in\bar{R}$};

\draw (r1) -- (r);
\end{tikzpicture}

\begin{tikzpicture}[scale=0.75]
\node (r) at (0.5,1.5) [fill,circle,inner sep=1pt] {}; 
\node (i) at (0,0) {};
\node (j) at (1.5,0) {};

\draw[draw=none] (i) -- (j) node [midway,below] {$T^r_{[0,\ldots,k]}$}; 
\draw[fill=gray!10] (j.center) -- (r.center) -- (i.center) -- cycle; 
\draw (r) node [above] {$r\in [0,\ldots,k]$};

\node (r1) at (3.5,1) [fill,circle,inner sep=1pt] {}; 
\node (i1) at (3,0) {};
\node (j1) at (4,0) {};

\draw[draw=none] (i1) -- (j1) node [midway,below] {$T^{r'}_{[k+1,\ldots,p-1]}$}; 
\draw[fill=gray!10] (j1.center) -- (r1.center) -- (i1.center) -- cycle; 
\draw (r1) node [anchor=south west] {$r'$};

\draw (r1) -- (r);

\node (r) at (9.5,1.5) [fill,circle,inner sep=1pt] {}; 
\node (i) at (9,0) {};
\node (j) at (10.5,0) {};

\draw[draw=none] (i) -- (j) node [midway,below] {$T^r_{[k+1,\ldots,p-1]}$}; 
\draw[fill=gray!10] (j.center) -- (r.center) -- (i.center) -- cycle; 
\draw (r) node [above] {$r\in [k+1,\ldots,p-1]$};

\node (r1) at (6.5,1) [fill,circle,inner sep=1pt] {}; 
\node (i1) at (6,0) {};
\node (j1) at (7,0) {};

\draw[draw=none] (i1) -- (j1) node [midway,below] {$T^{r'}_{[0,k]}$}; 
\draw[fill=gray!10] (j1.center) -- (r1.center) -- (i1.center) -- cycle; 
\draw (r1) node [anchor=south west] {$r'$};

\draw (r1) -- (r);
\end{tikzpicture}
\end{center}
\caption{Structure of optimal completion time, non-ordered and
  ordered, gather and scatter trees.  The communication time (for
  gather, in the direction of the root $r$) is
  $\alpha_{r'r}+\-\beta_{r'r}\size(T^{r'}_{\bar{R}})$, and
  $\alpha_{r'r}+\-\beta_{r'r}\size(T^{r'}_{[k+1,\ldots,p-1]})$ or
  $\alpha_{r'r}+\-\beta_{r'r}\size(T^{r'}_{[0,\ldots,k,]})$,
  respectively, which has to be paid in addition to the respective
  times $\max(\cost(T_R),\cost(T_{\bar{R}}))$ and
  $\max(\cost(T_{[0,\ldots,k]}),\cost(T_{[k+1,\ldots,p-1]}))$ for
  completing the subtrees. For the ordered trees, the cases where
  $r\in[0,\ldots,k]$ and $r\in[k+1,\ldots,p-1]$ both have to be
  considered.}
\label{fig:trees}
\end{figure}
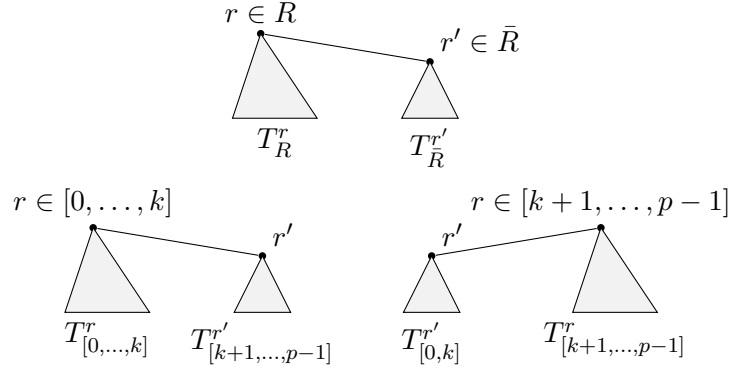

\begin{figure*}
\begin{eqnarray}
\cost(T^r_{[r]}) & = & 0 \nonumber \\
\cost(T^r_{[i,\ldots,j]}) & = &
\max(\cost(T^r_{[i,\ldots,k]}),\cost(T^{r'}_{[k+1,\ldots,j]})) + 
\alpha_{r'r}+\beta_{r'r}\size(T^{r'}_{[k+1,\ldots,j]}) \label{eq:orderleft} \\
\cost(T^r_{[i,\ldots,j]}) & = &
\max(\cost(T^{r'}_{[i,\ldots,k]}),\cost(T^{r}_{[k+1,\ldots,j]})) + 
\alpha_{r'r}+\beta_{r'r}\size(T^{r'}_{[i,\ldots,k]}) \label{eq:orderright} \\
\cost(T^r_{[r,\ldots,j]}) & = &
\max(\gamma_rm_r,\cost(T^{r'}_{[r+1,\ldots,j]})) + 
\alpha_{r'r}+\beta_{r'r}\size(T^{r'}_{[r+1,\ldots,j]}) \label{eq:copyleft}\\
\cost(T^r_{[i,\ldots,r]}) & = &
\cost(T^{r'}_{[i,\ldots,r-1]}) + 
\alpha_{r'r}+\beta_{r'r}\size(T^{r'}_{[i,\ldots,r-1]})+\gamma_rm_r \label{eq:copyright}
\end{eqnarray}
\caption{Equations characterizing the completion time of ordered
  gather and scatter trees over a range of processors $[i,\ldots,j]$
  rooted at processor $r\in[i,\ldots,j]$.}
\label{fig:optimalordered}
\end{figure*}

We now show that optimal, smallest completion time, ordered gather and
scatter trees can be constructed in polynomial trees.

\subsection{Characterizing tree completion times}

The following observation express the tree completion times more
concisely and is crucial for all following algorithms and
results. Propositions~\ref{prop:optimal} and~\ref{prop:optimalordered}
both assume one-ported communication, and are difficult to extend to
more communication ports.

\begin{proposition}
\label{prop:optimal}
An optimal completion time, non-ordered communication tree for an
irregular gather (or scatter) operation over a set $P$ (of at least
two) one-ported processors consists in a subtree $T^r_R$ rooted at
some processor $r\in R$ and a subtree $T^{r'}_{\bar{R}}$ rooted at
some other processor $r'\in \bar{R}$ where $R$ and $\bar{R}$ is a
partition of $P$ with communication between processors $r$ and $r'$
that minimizes $\cost(T^r_P)$ defined by the following equations.
\begin{eqnarray}
\cost(T^r_{\{r\}}) & = & 0 \nonumber \\
\cost(T^r_P) & = & \max(\cost(T^r_R),\cost(T^{r'}_{\bar{R}})) + 
\alpha_{r'r}+\beta_{r'r}\size(T^{r'}_{\bar{R}}) \label{eq:later}\\
\underset{R=\{r\}}{\cost(T^r_P)} & = & \max(\gamma_rm_r,\cost(T^{r'}_{\bar{R}})) +
\alpha_{r'r}+\beta_{r'r}\size(T^{r'}_{\bar{R}}) \label{eq:local}
\end{eqnarray}
\end{proposition}
\begin{proof}
The partition of $P$ into $R$ and $\bar{R}$ in
Equation~(\ref{eq:later}) defines the last subtree in the sequence of
subtrees of $T^r_R$ as in Equation~(\ref{eq:commcost}) of
Definition~\ref{def:treecost}, and assigns the same
cost. Equation~(\ref{eq:local}) defines the place in the sequence of
subtrees of $r$ for performing the local copy by partitioning $P$ into
$R=\{r\}$ and $\bar{R}=R\setminus\{r\}$ and corresponds to
Equation~(\ref{eq:localcopy}). This always puts the subtrees over
$\bar{R}$ after the local copy in the sequence, such that completion
of $T^{r'}_{\bar{R}}$ can be done concurrently with the local copy of
cost $\gamma_rm_r$. An optimal, non-ordered tree must have this
structure, since the case where communication is done first and then
the local copy (which cannot be overlapped) would have cost
$\cost(T^{r'}_{\bar{R}})+\alpha_{r'r}+\beta_{r'r}\size(T^{r'}_{\bar{R}})+\gamma_rm_r$
which is larger than $\max(\gamma_rm_r,\cost(T^{r'}_{\bar{R}})) +
\alpha_{r'r}+\beta_{r'r}\size(T^{r'}_{\bar{R}})$. A tree that is
optimal according to the proposition will therefore also be optimal
according to Definition~\ref{def:treecost}, and conversely.
\end{proof}

The importance of Proposition~\ref{prop:optimal} is in characterizing
optimal completion time trees.
\begin{corollary}
  \label{corr:dynprog}
In the one-ported transmission cost model, optimal completion time
gather and scatter trees exhibit optimal substructure.
\end{corollary}
\begin{proof}
  In order for the completion time of Equation~(\ref{eq:later}) to be
  optimal, both subtrees $T^r_R$ and $T^{r'}_{\bar{R}}$ must be
  optimal; if not a possibly better completion time tree could be found.
\end{proof}
  
Ordered and non-ordered trees are shown in Figure~\ref{fig:trees} to
illustrate Proposition~\ref{prop:optimal} and the following
Proposition~\ref{prop:optimalordered}.

Corollary~\ref{corr:dynprog} states that the problem of constructing
optimal gather and scatter trees can be solved by dynamic
programming. Unfortunately, Proposition~\ref{prop:optimal} seems to imply that
all possible partitions of the set of processors into the subsets $R$
and $\bar{R}$ need to be considered, and for non-ordered trees, this
might indeed be so as Section~\ref{sec:hardness} shows. For (strongly)
ordered trees where $P=[i,i+1,\ldots,j]$ is a consecutive list of
processors, however, only $j-i$ partitions have to be considered,
namely $R=[i,\ldots,k]$ and $\bar{R}=[k+1,\ldots,j]$ for $k=i,\ldots,j-1$.

\begin{proposition}
\label{prop:optimalordered}
An optimal, least cost, ordered communication tree for an irregular
gather or scatter problem over a list of processors
$P=[i,i+1,\ldots,j]$ with at least two processors consists in ordered
subtrees over $[i,\ldots,k]$ and $[k+1,\ldots,j]$ for some $k, i\leq
k<j$ with roots $r\in[i,\ldots,k]$ and $r'\in[k+1,\ldots,j]$, or
$r\in[k+1,\ldots,j]$ and $r'\in[i,\ldots,k]$ that minimizes
$\cost(T^r_P)$ defined by the equations given in
Figure~\ref{fig:optimalordered}.
\end{proposition}
\begin{proof}
The proposition follows by specialization of
Proposition~\ref{prop:optimal}, keeping track of whether the root is
in the list of processors $[i,\ldots,k]$ or $[k+1,\ldots,j]$.
Equations~(\ref{eq:orderleft}) and~(\ref{eq:orderright})
correspond to Equation~(\ref{eq:later}).
Equations~(\ref{eq:copyleft}) and~(\ref{eq:copyright}) correspond to
Equation~(\ref{eq:local}) but are asymmetric because of the strong
ordering constraints.  When $r\in[r,\ldots,j]$ is the first processor,
the local copy can be done concurrently with completing the tree
$T^{r'}_{[r+1,\ldots,j]}$ after which communication is paid for. When
on the other hand $r\in[i,\ldots,r]$ is the last processor, the data
blocks $[m_i,\ldots,m_{r-1}]$ must be communicated first after which
the local copy can take place.
\end{proof}

\begin{figure*}
  \begin{small}
  \begin{eqnarray*}
  C[i,i,i] & = & 0 \\
  \underset{i<j}{C[i,j,r]} & = & \left\{\begin{array}{cl}
  \min_{i+1\leq r'<j}[\max(\gamma_r m_r,C[i+1,j,r'])+\comm_{r'r}(S[i+1,j])] & 
r=i \\
  \min_{i\leq k<j}\min_{k+1\leq r'\leq j}[\max(C[i,k,r],C[k+1,j,r']) +
    \comm_{r'r}(S[k+1,j])] & 
r\in [i+1,\ldots,k] \\
  \min_{i\leq k<j}\min_{i\leq r'\leq k}[\max(C[i,k,r'],C[k+1,j,r]) +
  \comm_{r'r}(S[i,k])] & 
r\in [k+1,\ldots,j-1] \\
\min_{i\leq r'<j-1}[C[i,j-1,r']+\comm_{r'r}(S[i,j-1])+\gamma_jm_j] &
r=j \\
\end{array}\right.
\end{eqnarray*}
  \end{small}
\caption{Dynamic programming equations for constructing optimal,
  ordered gather under a non-homogeneous communication cost model
  $\comm_{ij}(m)=\alpha_{ij}+\beta_{ij}m$ where communication is from
  a non-root processor $r$ to root processor $r$. For the scatter
  tree, communication cost from root processor $r$ to non-root $r'$ is
  incurred instead.}
\label{fig:optimalequations}
\end{figure*}

\subsubsection*{Discussion}

Proposition~\ref{prop:optimal} and Corollary~\ref{corr:dynprog} do not
seem to extend to the case with $k>1$ communication ports. An optimal
completion time $k$-ported tree is not described by partitioning $P$
into $k+1$ optimal subtrees where $k$ ports of the root in one subtree
gathers or scatters from or to $k$ other subtrees. The $k$ subtrees
must all have completed (as determined by the slowest port of the root
processors in these trees) before communication and be optimal, but
the one tree with the root does not have to be optimal. A slow port at
the root could communicate with a fast subtree, leading to an overall
faster completion time. Constructing optimal trees for the $k>1$
ported case seems more difficult than with only one port, and it is
not known (to the author) whether this problem can be solved as
efficiently.

\subsection{Dynamic programming algorithms}

Proposition~\ref{prop:optimalordered} and Corollary~\ref{corr:dynprog}
show that the problem of finding optimal, ordered gather and scatter
trees for the one-ported communication model can be solved efficiently
by dynamic programming by building up optimal subtrees for larger and
larger ranges of processors $[i,\ldots,j]$. Using the optimality
criteria from Proposition~\ref{prop:optimalordered}, we now develop
the dynamic programming equations for both non-homogeneous and
homogeneous communication costs.

The size of a tree $T^r_{[i,j]}$ over a range of processors
$[i,\ldots,j]$ can be computed in constant time as
$\size(T^r_{[i,\ldots,j]}) = S[j]-S[i-1]$ from precomputed prefix sums
$S[i] = \sum_{j=0}^{i}m_j$ (with $S[-1]=0$).  The completion time of
an ordered tree over a range of processors $[i,\ldots,j]$ is
determined by the completion times of smaller, optimal subtrees over
ranges $[i,\ldots,k]$ and $[k+1,\ldots,j]$ and the communication time
for transmitting between the subtrees. With non-homogeneous
communication costs, transmission times can be different between
different pairs of processors. We let $C[i,j,r]$ denote the completion
time of an optimal gather or scatter tree over processors in the range
$[i,\ldots,j]$ rooted at processor $r\in[i,\ldots,j]$, and
$\comm_{ij}(m)=\alpha_{ij}+\beta_{ij}m$ the cost of communicating a
block of size $m$ in the non-homogeneous transmission cost model, with
in addition $\comm_{ij}(0) = 0$ meaning no communication when no
data. Now, $C[i,j,r]$ can be computed as given by the equations in
Figure~\ref{fig:optimalequations} which follow directly from
Proposition~\ref{prop:optimalordered}, given that the required values
$C[i,k,r]$ and $C[k+1,j,r']$ are available. There is no cost for
singleton trees $C[i,i,i]$. If the root $r$ is the first processor
$i$, there is a local copy cost, concurrently with which the other
tree over processors $[i+1,\ldots,j]$ can be completed, with the
additional cost of the communication between processor $i$ and the
root $r'$ in the best tree over the processors
$[i+1,\ldots,j]$. Similarly if the root $r$ is the last processor $j$
in the range, with the difference that the local copy in this case can
be done only after the communication has taken place. Otherwise, there
are two cases, depending on whether $r\in[i,\ldots,k]$ or
$r\in[k+1,\ldots,j]$. In both, the larger of the completion times for
the best two subtrees have to be paid, in addition to communication
costs between the root processors $r$ and $r'$.  The cost of an
optimal gather or scatter tree over the whole range of processors
$[0,p-1]$ with root $r$ is given by the table entry $C[0,p-1,r]$. If
the overall best tree with root processor $r$ chosen by the algorithm
is needed, this is given by $\min_{0\leq r<p}C[0,p-1,r]$.

The values $C[i,j,r]$ can be maintained in a three-dimensional table
$C$ that can be computed bottom up, step by step increasing the size
of the processor ranges $[i,\ldots,j]$. Each such step requires a pass
over all $i$, all $k,i\leq k<j$, and all $r'\in [i,\ldots,j]$, for a
total of $O(p^4)$ operations over all range increasing
steps. Computing the sizes $S[i,j]$ is done on the fly in constant
time from the precomputed prefix sums. The actual trees can be
constructed by keeping an additional two-dimensional table
$\theroot[i,j]$ that keeps track of the subtree root $r'$ chosen for
the processor range $[i,\ldots,j]$. The trees constructed are not necessarily
binomial, or binary; each subtree may have an arbitrary degree between
$0$ and $p-1$. We say that optimal trees have \emph{variable degree},
and have argued for the following theorem.

\begin{theorem}
\label{thm:optimalnonhomo}
Completion time optimal, variable degree gather and scatter trees for the
irregular gather and scatter problems on $p$ fully connected, one-ported
processors under a linear-time, \emph{non-homogeneous transmission
  cost model} can be computed in $O(p^4)$ operations using $O(p^3)$
space.
\end{theorem}
\begin{proof}
Optimality follows from Proposition~\ref{prop:optimalordered}, and the
complexity from the dynamic programming construction of the
three-dimensional table $C$. Each part of the equations in
Figure~\ref{fig:optimalequations} takes at most $O(p^2)$ time, and the
table can be constructed by two nested loops of at most $p$ iterations
each.
\end{proof}

\begin{figure*}
  \begin{small}
\begin{eqnarray*}
  C[i,i] & = & 0 \\
  C[i,i+1] & = & \min\left\{
  \begin{array}{l}
    \max(\gamma m_i,C[i+1,i+1])+\comm(S[i+1,i+1])
    \\
    C[i,i]+\comm(S[i,i])+\gamma m_j
    \\
  \end{array}
  \right\} \\
  \underset{i+1<j}{C[i,j]} & = & \min_{i\leq k<j}[\max(C[i,k],C[k+1,j])+\min(\comm(S[i,k]),\comm(S[k+1,j]))]
\end{eqnarray*}
  \end{small}
\caption{Dynamic programming equations for constructing optimal,
  ordered gather and scatter trees under a homogeneous communication
  cost model $\comm(m)=\alpha+\beta m$.}
\label{fig:homooptimal}
\end{figure*}

For non-homogeneous communication systems, when two subtrees over a
range of processors communicate, it is necessary to consider all
possible best rooted trees over the range, since the communication
cost from the different roots may differ. For homogeneous systems,
where the $\alpha$ and $\beta$ parameters are the same for all pairs
of processors, this is not the case, and it suffices to keep track of
the overall best rooted tree for each processor range $[i,\ldots,j]$
instead of each possible root in $[i,\ldots,j]$.  This leads to the
simplified, dynamic programming equations given in
Figure~\ref{fig:homooptimal} which find a least cost gather or scatter
tree for a best possible root. Since only a best possible root needs
to be maintained for each range, this lowers the overall complexity by
a factor of $p$ to $O(p^3)$ operations, leading to the following
theorem.

\begin{theorem}
\label{thm:optimalhomo}
Completion optimal variable degree gather and scatter trees for the
irregular gather and scatter problems on $p$ fully connected, one-ported
processors under a linear-time, homogeneous transmission cost model
can be computed in $O(p^3)$ operations using $O(p^2)$ space.
\end{theorem}
\begin{proof}
The correctness of the equations in Figure~\ref{fig:homooptimal}
follow from Proposition~\ref{prop:optimalordered}.  The
two-dimensional $C$ table is constructed by a standard, dynamic
programming algorithm in $O(p^3)$ operations.
\end{proof}

We note that by explicitly keeping track of the root chosen for each
range $[i,\ldots,j]$ in a separate, two-dimensional table
$\theroot[i,j]$, the computed trees can easily be reconstructed as is
a standard dynamic programming technique, see,
\eg,~\cite{CormenLeisersonRivestStein09}. For computing a tree rooted
at an externally given root $r$, the equations have to be modified so
that communication is always with this root for processor ranges
$[i,\ldots,j]$ with $r\in [i,\ldots,j]$. Also this is straightforward.

\begin{figure*}
  \begin{small}
  \begin{eqnarray*}
  C[i,i] & = & 0 \\
  \underset{i<j}{C[i,j]} & = & \min\left\{
  \begin{array}{l}
    \min_{i<k<j}\left[\max(\max(\gamma m_i,C[i+1,k])+\comm(S[i+1,k]),C[k+1,j])+\comm(S[k+1,j])\right]
    \\
    \min_{i<k<j}\left[\max(C[i,k-1]+\comm(S[i,k-1])+\gamma m_k,C[k+1,j])+\comm(S[k+1,j])\right]
    \\ 
    \min_{i\leq k<j-1}\left[\max(C[i,k]+\comm(S[i,k]),C[k+1,j-1])+\comm(S[k+1,j-1])+\gamma m_j\right]
    \\ 
    \max(\gamma m_i,C[i+1,j])+\comm(S[i+1,j])
    \\
    C[i,j-1]+\comm(S[i,j-1])+\gamma m_i
    \\
  \end{array}
  \right\}
\end{eqnarray*}
  \end{small}
\caption{Dynamic programming equations for constructing cost optimal,
  ordered binary trees under a homogeneous transmission cost model
  $\comm(m)=\alpha+\beta m$.}
\label{fig:binaryequations}
\end{figure*}

\subsection{Dynamic programming for other trees}

Binary (or other fixed-degree trees) are sometimes used for
implementing rooted collective operations like broadcast and
reduction. Dynamic programming can likewise be used to compute optimal
binary trees for these problems. For completeness (and because this
could be relevant for, \eg, sparse reduction problems, or for
exploring the structure of optimal trees) we state the corresponding
dynamic programming equations. There are five cases to consider. The
root of an ordered, binary tree is either ``at the left'', ``at the
right'' or somewhere ``in the middle'', or the root has only one
child, either left or right. The equations assuming a homogeneous cost
model are given in Figure~\ref{fig:binaryequations}. Also here, a best
possible root is chosen by the algorithm; it is easy to modify to the
case where the root is externally given; here only the last extension
step filling in table entry $C[0,p-1]$ needs to be adapted to sending
to the fixed root $r$ and choose the best subtrees out of the five
possible cases.  This gives the following theorem.
\begin{theorem}
\label{thm:binaryordered}
Completion time optimal binary gather and scatter trees for the
irregular gather and scatter problems on $p$ fully connected, one-ported
processors under a linear-time, homogeneous transmission cost model
can be computed in $O(p^3)$ operations ($O(p^4)$ for non-homogeneous
communication costs) and $O(p^2)$ space ($O(p^3)$ for non-homogeneous
communication costs).
\end{theorem}

By the same techniques, optimal, ordered trees for broadcast and
reduction can likewise be computed. The only difference is that the
size for each node in such trees are the same, namely $S[i,j]=m$ for
all $i\leq j$. Note that the trees computed in this way are special in
the sense that the subtrees of interior nodes span consecutive
ranges of processors $[i,\ldots,j]$. As Section~\ref{sec:hardness}
will show, not all optimal broadcast trees have this
structure. However, this property is useful for reduction operations,
where binary reduction operations may have to be performed in
processor rank order (unless the operation is commutative); MPI for
instance has such constraints.

\begin{theorem}
Cost-optimal, ordered broadcast and reduction trees on $p$ fully
connected, one-ported processors under a linear-time, non-homogeneous
transmission cost model can be computed in $O(p^4)$ operations and
$O(p^3)$ space (respectively $O(p^3)$ time and $O(p^2)$ space under
homogeneous costs).
\end{theorem}

\subsubsection*{Discussion}

The dynamic programming algorithms are hardly practically relevant,
unless trees can be precomputed and reused many times in persistent
communications, or unless problems are so large that $m$ is
$\Omega(p^4)$ (or $\Omega(p^3)$) and actual communication time offsets
the tree construction time. Furthermore, the constructions are
offline, and require full knowledge of the block sizes for all
processors. In MPI and other interfaces, only the root processor has
this information. A schedule could be computed at the root, and sent
to the other processors (with at least a linear time communication
overhead), an idea that was attempted in~\cite{Traff04:gatscat}.

\section{Hardness of the Non-Ordered Constructions}
\label{sec:hardness}

The ordering constraint on communication trees made polynomial time
constructions of optimal gather and scatter trees possible. The
following theorems show that finding cost-optimal trees when subtrees
are not required to be ordered is a different, harder problem. The
first and second are easy observations.

\begin{theorem}
Constructing optimal, minimum completion time broadcast trees in the
non-homogeneous, fully-connected, one-ported, linear transmission cost
model is NP-hard.
\end{theorem}
This is a simple observation and reduction from the Minimum Broadcast
Time problem, see~\cite[ND49]{GareyJohnson79}
and~\cite{SlaterCockayneHedetniemi81}. The Minimum Broadcast Time
problem is, for a given unweighted, undirected graph $G$ and root
vertex $u$ to find the minimum number of communication rounds required
to broadcast a unit size message from $u$ to all other vertices in
$G$, where in each communication round, a vertex can send a message to
some other vertex along an edge of $G$. We take $\alpha_{uv}=1$ for
all edges $uv$ of $G$, and $\alpha_{u'v'}=\infty$ for edges $u'v'$ not
in $G$, $m=1$ and all $\beta_{uv}=1$. An optimal broadcast tree in the
non-homogeneous cost model would be a solution to the Minimum
Broadcast Time problem, and vice versa.

As indicated in Section~\ref{sec:lowerbounds} with $k$-ported
communication, $k>1$, an optimal algorithm would entail solving a hard
packing problem.

\begin{theorem}
Finding an optimal solution to the irregular gather and scatter
problems in the $k$-ported, homogeneous, fully connected, linear
transmission cost model is NP-hard for $k>1$.
\end{theorem}
This is a simple reduction from Multiprocessor
Scheduling~\cite[SS8]{GareyJohnson79}. Given an instance of the
scheduling problem, we construct a gather scatter operation over $p$
processors where $p$ is equal to the number of jobs, plus one extra
root processor $r$. We take $\alpha_{ij}=0$ and $\beta_{ij}=1$ for all
processor pairs $i,j$, with $k$ equal to the number of machines, and data
block $m_i$ equal to length of the $i$th job and $m_r=0$. With no
latency, optimal gather and scatter trees are stars rooted at $r$, and the
completion time of the gather/scatter algorithm which is the time for
the last communication port to finish corresponds directly to the
scheduling deadline.

More interestingly, and less obvious, even in the one-ported,
homogeneous communication cost model, finding non-ordered gather and
scatter trees remains hard. The details are worked out in the proof.

\begin{theorem}
\label{thm:hardness}
Finding non-ordered, cost-optimal communication trees for irregular
gather and scatter operations in the homogeneous, fully connected,
one-ported, linear transmission cost model is NP-hard.
\end{theorem}

\begin{proof}
The claim is by reduction from the PARTITION
problem~\cite[SP12]{GareyJohnson79} to the problem of constructing an
optimal gather tree.

An instance of the PARTITION problem is a set of positive integers,
$m_i, i=0,\ldots,p-1, m_i>0$, with $m=\sum_{i=0}^{p-1}m_i$ even. The
problem is to determine whether there is a subset $R$ of
$\{0,\ldots,p-1\}$ with $\sum_{i\in R}m_i=m/2$.  The problem is
trivial if there is some $m_i$ with $m_i\geq m/2$, so we assume that
for all $m_i, m_i<m/2$. For any subset $R$ of $\{0,\ldots,p-1\}$ we
let $\bar{R}$ denote the complement $\{0,\ldots,p-1\}\setminus R$.

For the cost model we take $\gamma=\beta=1$, meaning that there is a
local copy cost for each processor for its own data block. We consider
gather tree constructions where the root processor is chosen by the
algorithm. By Proposition~\ref{prop:linear} we know that under these
assumptions, an ordered gather tree with cost $\ceiling{\log_2
  p}\alpha+\sum_{i=0}^{p-1}m_i$ over the integers in the PARTITION
instance can be computed in $O(p)$ time steps.

With $p$ the number of integers in the PARTITION instance, we define
$\alpha=1/p$ such that $\alpha p = 1$.  From the given PARTITION
instance, we construct an irregular gather problem for $p+2$
processors with block sizes $m'_i,0\ldots,p-1,p,p+1$ defined as
follows. First, $m'_i=m_i+\alpha$ for $i=0,\ldots,p-1$. Two
additional, large data blocks $m'_p$ and $m'_{p+1}$ are introduced for
the purpose of hiding the cost of gathering in certain subtrees. The
sizes of these blocks are $m'_p=m/2$ and $m'_{p+1}=m/2+1$. The size of
the constructed gather problem is 
\begin{displaymath}
m'=\sum_{i=0}^{p+1}m'_i=\sum_{i=0}^{p-1}(m_i+\alpha)+m/2+m/2+1=2m+2.
\end{displaymath}

\begin{figure}
\begin{center}
\begin{tikzpicture}[scale=0.75]
\draw[fill=gray!10] (0,0) node [below] {$m'_{p+1}$} rectangle (0.5,4.5);
\node (m) at (0.25,4) [fill,circle,inner sep=1pt] {};

\node (r) at (2,3) [fill,circle,inner sep=1pt] {}; 
\node (i) at (1.5,0) {};
\node (j) at (3,0) {};

\draw[draw=none] (i) -- (j) node [midway,below] {$T_R$};
\draw[fill=gray!10] (j.center) -- (r.center) -- (i.center) -- cycle;

\draw (r) -- (m) node [midway,right,below] {$\alpha$};

\draw[fill=gray!10] (4,0) node [below] {$m'_{p}$} rectangle (4.5,4);
\node (m1) at (4.25,3.5) [fill,circle,inner sep=1pt] {};

\node (r1) at (6,2) [fill,circle,inner sep=1pt] {}; 
\node (i1) at (5.5,0) {};
\node (j1) at (6.5,0) {};

\draw[fill=gray!10] (j1.center) -- (r1.center) -- (i1.center) -- cycle;


\draw (r1) -- (m1) node [midway,right,below] {$\alpha$};

\node (r2) at (8,1) [fill,circle,inner sep=1pt] {}; 
\node (i2) at (7.5,0) {};
\node (j2) at (8.5,0) {};

\draw[fill=gray!10] (j2.center) -- (r2.center) -- (i2.center) -- cycle;

\draw (r2) -- (m1) node [midway,right,above] {$\alpha$};

\node[draw=none,fill=none,below] at (7,0) {$\bar{R}$};

\draw (m) -- (m1) node [midway,right,above] {$\alpha$};
\end{tikzpicture}
\end{center}
\caption{An optimal gather tree corresponding to a solution of a
  PARTITION instance $m_i, i=0,\ldots,p-1$ with
  $m=\sum_{i=0}^{p-1}m_i$. The input for the gather problem is
  $m'_i,i=0,\ldots,p-1,p,p+1$ with $m'_i=m_i+\alpha$ for
  $i=0,\ldots,p-1$, $m'_p=m/2$ and $m'_{p+1}=m/2+1$, with $\alpha=1/p$
  and $\gamma=\beta=1$. The cost of an optimal gather tree is
  $2m+2+2\alpha$ if and only if there a solution to the PARTITION
  instance. The sets $R$ and $\bar{R}$ is the partition of the indices
  $\{0,\ldots,p-1\}$. The proof argues that any other tree will have
  higher cost.}
\label{fig:partition}
\end{figure}
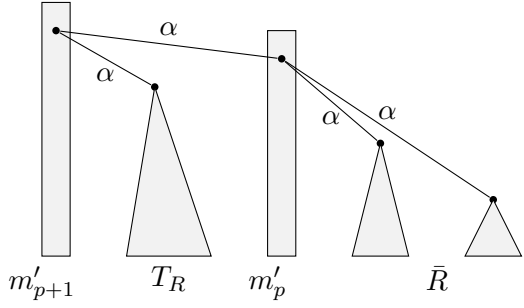

The claim is that there is a solution to the PARTITION instance with
partition into subsets $R$ and $\bar{R}$ if and only if the optimal,
lowest completion time gather tree completes in time
$2m+2+2\alpha$. Furthermore, this optimal time gather tree has the
structure shown in Figure~\ref{fig:partition}, in which the processors
$p+1$ and $p$ having the two large blocks will be the roots in two
subtrees, each gathering blocks from the processors for the two sets
$R$ and $\bar{R}$ and processor $p+1$ finally receiving all blocks
gathered by processor $p$. Recall the optimality criterion of
Proposition~\ref{prop:optimal}, which states that the cost of a gather
tree is determined by the most expensive of two subtrees plus the cost
of communication between two subtrees. Also note that with
$\gamma=\beta=1$, any gather tree will take at least $m'=2m+2$ time to
complete. Finally, note that each non-leaf processor $i,0\leq i<p+2$
in a non-ordered algorithm can do the local copy of its block $m'_i$
at the point where it causes the least overall cost, that is
concurrently with the construction of any one of its subtrees.

For the ``if'' part, first assume that $R,\bar{R}$ is a solution to
the PARTITION instance. Let $R$ be the subset with the smallest number
of elements (denoted by $|R|$) which is at most $p/2$ such that
$p/2+\ceiling{\log_2 |R|}<p$. By Proposition~\ref{prop:linear}, an
optimal gather tree for this subset has cost at most
\begin{eqnarray*}
\ceiling{\log_2 |R|}\alpha+\sum_{i\in R}m'_i & = & 
\ceiling{\log_2 |R|}\alpha+\sum_{i\in R}(m_i+\alpha) \\
& = & \ceiling{\log_2 |R|}\alpha+|R|\alpha+\sum_{i\in R}m_i \\
& = & \ceiling{\log_2 |R|}\alpha+|R|\alpha+m/2
\end{eqnarray*} 
which is smaller than the cost $\gamma m'_{p+1}=m/2+1$ for the local
copy of the large block at processor $p+1$. The subtree over $R$ can
therefore be constructed concurrently with the local copy, and the
blocks from the subtree sent to processor $p+1$, for a total cost of
$m/2+1+\alpha+m/2+|R|\alpha=m+1+\alpha+|R|\alpha$.

The elements in $\bar{R}$ can be gathered to processor $p$ in two
steps in time $m+|\bar{R}|\alpha+2\alpha$ as indicated in
Figure~\ref{fig:partition}. The processors for the elements in
$\bar{R}$ are split into two non-empty parts, such that the gather
time for each is less than $m/2$, the time for processor $p$ to
locally copy its block of size $m'_p=m/2$. The blocks for the two
parts are sent to processor $p$ in two operations incurring two times
the latency $\alpha$ and a total time of $m+|\bar{R}|\alpha+2\alpha$.
Since $m+|\bar{R}|\alpha+2\alpha < m+1+\alpha+|R|\alpha$, the two
trees rooted at processor $p+1$ and processor $p$ can be constructed
concurrently, with total time for completing the gathering at root
processor $p+1$ being
\begin{displaymath}
m+1+\alpha+|R|\alpha+\alpha+m+|\bar{R}|\alpha = 2m+2+2\alpha
\end{displaymath}
as claimed, recalling that $|R|\alpha+|\bar{R}|\alpha=\alpha p=1$.

For the ``only if'' part, we argue that any other tree takes longer
than $2m+2+2\alpha$ to complete gathering. Therefore, if there is no
solution to the PARTITION instance, the optimal gather tree has a
different structure than shown in Figure~\ref{fig:partition} and takes
longer. Since each send operation to the root adds at least an
$\alpha$ term, we do not have to consider trees where the root has
more than two children. Also, the processors $p$ and $p+1$ having the
large blocks must be subtree roots, since sending a large block, say,
$m'_{p+1}=m/2+1$ would incur at least
$\alpha+m'_{p+1}$ extra time instead of only $m'_{p+1}$ for a local copy when
processor $p+1$ is a local root. We therefore only have to consider
trees with the structure shown in Figure~\ref{fig:partition} and
Figure~\ref{fig:nopartition}. The case where the root processor $p+1$
receives from only one child would take much longer, resulting
from first gathering all elements from $R\cup\bar{R}$ at processor $p$
which takes time at least $m'_p+m+1+\alpha$ and then sending the
$m/2+m+1$ units to processor $p+1$. This is in total at least
$3m+2+2\alpha$.

We first consider the case where the block sizes of the two subsets
$R$ and $\bar{R}$ are not balanced, and argue that the tree shown in
Figure~\ref{fig:partition} has completion time larger than
$2m+2+2\alpha$.

Assume that $\sum_{i\in R}m_i>m/2$. Since trivial solutions to the
partition problem are excluded, $R$ consists of at least two
processors, and the gather time of a tree rooted at some processor in $R$
is therefore at least $\alpha+\sum_{i\in R}m'_i$.  Since each $m_i\geq
1$ it follows that $\sum_{i\in R}m_i\geq m/2+1=m'_{p+1}$, therefore
the local copy at processor $p+1$ is best done concurrently with
gathering in $R$ in order to balance the two terms in
Proposition~\ref{prop:optimal}.  The time for processor $p+1$ to gather
from the root of $R$ is therefore at least $2\alpha+\sum_{i\in
  R}m'_i$. Now gathering from the root of $\bar{R}$ adds at least
$\alpha+m'_p+\sum_{i\in\bar{R}}m'_i$ to the total time, which is thus
at least $3\alpha+2m+2$.

If on the other hand $\sum_{i\in R}m_i<m/2$, and consequently
$\sum_{i\in\bar{R}}m_i\geq m/2+1=m'_{p+1}$, the subtree rooted at
processor $p$ would have completion time larger than
$m'_{p+1}+\alpha+\sum_{i\in R}m'_i$ for a total time after gathering
at processor $p+1$ of at least $2m'_{p+1}+2\alpha+\sum_{i\in
  R}m'_i+\sum_{i\in \bar{R}}m'_i = 2m+3+2\alpha$.

\begin{figure}
\begin{center}
\begin{tikzpicture}[scale=0.75]
\draw (0,0) node [below] {$m'_{p+1}$} rectangle (0.5,4.5) [fill=red!10];
\node (m) at (0.25,4) [fill,circle,inner sep=1pt] {};

\draw (1.5,0) node [below] {$m'_{p}$} rectangle (2,4) [fill=red!10];
\node (m1) at (1.75,3.5) [fill,circle,inner sep=1pt] {};

\node (r) at (4.5,3) [fill,circle,inner sep=1pt] {}; 
\node (i) at (3,0) {};
\node (j) at (6.5,0) {};

\draw[draw=none] (i) -- (j) node [midway,below] {$T_{R\cup\bar{R}}$};
\draw [fill=red!10] (j.center) -- (r.center) -- (i.center) -- cycle;

\draw (r) -- (m) node [midway,right,above] {$\alpha$};

\draw (m) -- (m1) node [midway,right,below] {$\alpha$};
\end{tikzpicture}
\end{center}
\caption{Extreme case gather tree for a for a PARTITION instance with
  cost larger than $2m+2+2\alpha$. By the choice of $\alpha$, the cost
  of the tree $T_{R\cup \bar{R}}$ is at least $m+1+2\alpha$ and larger
  that $m+1+\alpha$ which is the cost of $m'_{p+1}$.}
\label{fig:nopartition}
\end{figure}
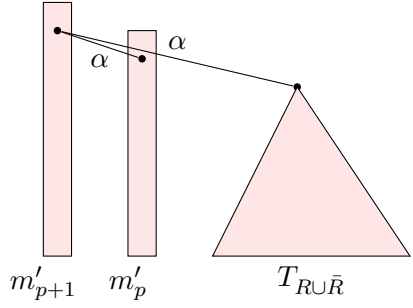

Also the structure shown in Figure~\ref{fig:nopartition} has larger
completion time than $2m+2+2\alpha$. Since $R\cup\bar{R}$ by the
non-triviality assumption has at least three elements, the cost of
gathering in a tree over $R\cup\bar{R}$ is at least $m+1+2\alpha$ as
shown in Lemma~\ref{lem:twochildren} stated below.  Therefore an
optimal tree rooted at $p+1$ which first receives block $m_p$ and
subsequently all elements of the PARTITION instance would take time at
least $2m+2+3\alpha$.

In summary, if there is no solution to the given PARTITION instance,
the completion time of an optimal gather tree is strictly larger than
$2m+2+2\alpha$, as claimed.
\end{proof}

For the proof, the following structural lemma was needed.
\begin{lemma}
\label{lem:twochildren}
Let $m_i,0\leq i<p$ be the block sizes of a gather problem of size
$m=\sum_{i=0}^{p-1}m_i$ over at least three processors $p$ and with
each $m_i$ a positive integer with $m_i<\floor{m/2}$. An optimal
gather tree has at least two children and cost at least $2\alpha+m$ in
a system with communication parameters $\gamma=\beta=1$, and $0<\alpha<1$ with
$\ceiling{\log_2 p}\alpha<1$.
\end{lemma}
\begin{proof}
Assume that some processor $r\in\{0,\ldots,p-1\}$ is chosen as root in
an optimal tree. Since $m_r<m/2$, and $\sum_{i=0,\ldots,r-1,r+1,\ldots
  p-1}m_i\geq m/2+1>m/2$, having the root receive all elements from
just one child would cost at least $m+2+\alpha$. An optimal tree will
therefore have at least two children, and cost at least $m+2\alpha$.
\end{proof}

In the proof of Theorem~\ref{thm:hardness}, we chose $\alpha<1$ with
$p\alpha=1$. The proof can easily be adopted for the case where only
integer costs are allowed for the model parameters
$\alpha,\beta,\gamma$. Simply choose $\alpha=1$, and construct the
gather problem with $m'_i=pm_i+1$, $m'_p=pm/2$ and
$m'_{p+1}=pm/2+p$. The claim is that an optimal gather tree has cost
$2pm+2p+2$ if an only if there is a solution to the PARTITION instance
given by the $m_i$ elements.

As an example, consider the PARTITION instance
$m_i=[2,3,3,3,4,4,4,6,6,13]$ with $p=10$ and $m=48$. This instance has
a solution with $R$ being the elements $[3,3,4,4,4,6,6]$, and an
optimal gather tree as constructed in the proof of
Theorem~\ref{thm:hardness} would complete in time $982=2\times 10\times
48+2\times 10+2$. However, considered as an ordered problem
$[2,3,3,3,4,4,4,6,6,13]$ with two extra elements $m_p=pm/2=240$ and
$m'_{p+1}=pm/2+p=250$ leads to a solution of cost $983$ corresponding
to one extra subtree. Thus, relaxing the ordering constraint can lead
to slightly better solutions, but finding the better tree is an
NP-hard problem as now proved in Theorem~\ref{thm:hardness}.

\section{Relative Quality of Gather Trees}
\label{sec:quality}

\begin{table*}[ht]
\caption{Gather tree completion times for low latency network,
  $p=2000,b=1000,r=\floor{p/2},\rho=5,\alpha=100,\beta=1,\gamma=1,0$. For
  the different tree constructions, both the cases where the root $r$
  is given externally (first row), and the case where a best root is
  chosen by the algorithm (second row) are shown. For the latter, the
  chosen root $(r)$ is given in brackets.}
\label{tab:res-lowlatency1000}
\begin{center}
\begin{tiny}
  \begin{tabular}{lrrrrrrrrrrr}
  \hline
$\gamma=1$ & $m$ & Linear & ($r$) & Binary & ($r$) & Oblivious & ($r$) & Adaptive & ($r$) & Optimal & ($r$) \\
 Same & 2000000 & 2199900 & & 3611600 & & 2001100 & & 2001100 & & 2001100 & \\
      & 2000000 & 2199900 & (0) & 3227500 & (0) & 2001100 & (0) & 2001100 & (1023) & 2001100 & (1998) \\
 Random & 1983668 & 2183568 & & 3587630 & & 2031225 & & 1994851 & & 1984868 & \\
      & 1983668 & 2183568 & (0) & 3198069 & (753) & 2021142 & (511) & 1984768 & (331) & 1984768 & (891) \\
 Random & 1983668 & 2183568 & & 4376141 & & 2339507 & & 2249565 & & 1984868 & \\
 decreasing     & 1983668 & 2183568 & (0) & 3193753 & (0) & 1984768 & (0) & 1984768 & (1) & 1984668 & (0) \\
 Random & 1983668 & 2183568 & & 3882776 & & 3023218 & & 2930348 & & 1984868 & \\
 increasing     & 1983668 & 2183568 & (0) & 3201309 & (1236) & 2077638 & (1999) & 1984768 & (1791) & 1984668 & (1998) \\
 Bucket & 2005668 & 2205568 & & 3613831 & & 2033113 & & 2013851 & & 2006868 & \\
      & 2005668 & 2205568 & (0) & 3234535 & (0) & 2033113 & (0) & 2006768 & (215) & 2006768 & (1135) \\
 Spikes & 2001600 & 2201500 & & 3569427 & & 2142672 & & 2037693 & & 2002900 & \\
      & 2001600 & 2201500 & (0) & 3219118 & (775) & 2132722 & (1999) & 2002700 & (1549) & 2002600 & (1198) \\
  Decreasing & 2003000 & 2202900 & & 4416588 & & 2353624 & & 2266244 & & 2004200 & \\
      & 2003000 & 2202900 & (0) & 3224561 & (0) & 2004100 & (0) & 2004100 & (1) & 2004000 & (0) \\
  Increasing & 2003000 & 2202900 & & 3917085 & & 3048580 & & 2955452 & & 2004200 & \\
      & 2003000 & 2202900 & (0) & 3233477 & (1999) & 2097228 & (1999) & 2004100 & (1791) & 2004000 & (1998) \\
  Alternating & 2000000 & 2199900 & & 3610600 & & 2001100 & & 2001100 & & 2001100 & \\
      & 2000000 & 2199900 & (0) & 3227500 & (0) & 2001100 & (0) & 2001100 & (1023) & 2001100 & (1998) \\
Skewed & 2001995 & 2201895 & & 4403490 & & 17203057 & & 4003090 & & 2003495 & \\
      & 2001995 & 2201895 & (0) & 2402295 & (0) & 2003095 & (0) & 2003095 & (3) & 2002295 & (2) \\
  Two blocks & 2000000 & 2000200 & & 3000200 & & 11001100 & & 3000200 & & 2000200 & \\
      & 2000000 & 2000100 & (0) & 2000100 & (0) & 2000100 & (0) & 2000100 & (1999) & 2000100 & (0) \\
\hline
$\gamma=0$ & $m$ & Linear & ($r$) & Binary & ($r$) & Oblivious & ($r$) & Adaptive & ($r$) & Optimal & ($r$) \\
 Same & 2000000 & 2198900 & & 3609600 & & 2000100 & & 2000100 & & 2000100 & \\
      & 2000000 & 2198900 & (0) & 3226500 & (0) & 2000100 & (0) & 2000100 & (1023) & 2000100 & (1998) \\
 Random & 1983668 & 2182874 & & 3579560 & & 2029413 & & 1992911 & & 1984074 & \\
      & 1983668 & 2181898 & (1996) & 3191893 & (0) & 2019330 & (511) & 1982828 & (330) & 1982768 & (891) \\
 Random & 1983668 & 2182589 & & 4373405 & & 2337510 & & 2247565 & & 1983889 & \\
 decreasing & 1983668 & 2181568 & (0) & 3192013 & (0) & 1982768 & (0) & 1982768 & (0) & 1982668 & (0) \\
 Random & 1983668 & 2182589 & & 3878419 & & 3021235 & & 2928563 & & 1983889 & \\
 increasing & 1983668 & 2181568 & (1999) & 3195849 & (1999) & 2075655 & (1999) & 1982983 & (1791) & 1982671 & (1998) \\
 Bucket & 2005668 & 2204375 & & 3608690 & & 2031871 & & 2012355 & & 2005675 & \\
      & 2005668 & 2204235 & (1999) & 3231178 & (1999) & 2031871 & (0) & 2005272 & (215) & 2005269 & (891) \\
 Spikes & 2001600 & 2201499 & & 3544637 & & 2137672 & & 2032693 & & 2002699 & \\
      & 2001600 & 2196500 & (1992) & 3194207 & (775) & 2127722 & (1999) & 1997700 & (1549) & 1997600 & (1198) \\
  Decreasing & 2003000 & 2201899 & & 4413708 & & 2351624 & & 2264243 & & 2003199 & \\
      & 2003000 & 2200899 & (0) & 3223075 & (0) & 2002099 & (0) & 2002099 & (0) & 2001999 & (0) \\
  Increasing & 2003000 & 2201898 & & 3912715 & & 3046595 & & 2953659 & & 2003198 & \\
      & 2003000 & 2200899 & (1999) & 3227925 & (1999) & 2095243 & (1999) & 2002307 & (1791) & 2002000 & (1998) \\
  Alternating & 2000000 & 2198400 & & 3604600 & & 2000600 & & 1999600 & & 1999600 & \\
      & 2000000 & 2198400 & (0) & 3223000 & (0) & 2000600 & (0) & 1999600 & (1022) & 1999600 & (1998) \\
Skewed & 2001995 & 2201894 & & 4003489 & & 16803057 & & 3603090 & & 2003294 & \\
      & 2001995 & 1801895 & (0) & 2002295 & (0) & 1603095 & (0) & 1603095 & (3) & 1602295 & (2) \\
  Two blocks & 2000000 & 2000200 & & 2000200 & & 11001100 & & 2000200 & & 2000200 & \\
      & 2000000 & 1000100 & (0) & 1000100 & (0) & 1000100 & (0) & 1000100 & (1999) & 1000100 & (0) \\
\hline
\end{tabular}
\end{tiny}
\end{center}
\end{table*}

Since the costs of computing optimal, ordered gather and scatter trees
are high, it makes sense to compare the completion times of optimal
trees to the completion times of other types of trees that can be
less expensively constructed. Such a comparison can also throw light
on the performance (problems) with simple constructions often used in
communication interfaces like, \eg, MPI. For this comparison, we now
focus on the irregular gather problems. Scatter trees will have the same
completion times.

\subsection{Model comparisons}

We have implemented the dynamic programming algorithms for
constructing ordered, varying degree trees as well as ordered, binary
trees for homogeneous communication networks from
Theorems~\ref{thm:optimalhomo} and~\ref{thm:binaryordered}.  In
addition, we build simple, linear-latency, star-shaped communication
trees, in which processors send directly to the root processor one
after the other, as well as standard, rank-ordered binomial
trees. These latter trees (linear and binomial) are
\emph{problem-oblivious} in the sense that the structure of the
communication tree is determined solely by $p$, the number of
processors, and by not the problem block sizes $m_i$. In addition, we
have implemented the algorithm for constructing problem (size and
distribution) aware binomial trees from
Proposition~\ref{prop:linear}. We call these trees
\emph{(problem-)adaptive}.  The algorithm constructs a binomially
structured tree that avoids waiting times by always letting the root
of the faster tree (smaller size) send its data to the root of an
adjacent tree with larger completion time. We are interested in seeing
how the adaptive algorithms fare, in particular, how far the
problem-adaptive binomial tree construction is from the optimal,
ordered algorithms. If not far, there is no reason to spend $O(p^3)$
operations (offline) in precomputing an optimal, ordered tree for a
given gather problem; if far, it would make sense to look for good,
parallel algorithms for computing the optimal gather and scatter
trees. The constructions are done offline, and we calculate the model
costs for homogeneous networks with chosen $\alpha,\beta,\gamma$
parameters.  All the implemented algorithms also explicitly construct
the ordered communication trees with the corresponding completion
times as explained\footnote{All implementations used compute
  completion times and the corresponding trees available at
  \url{par.tuwien.ac.at/Downloads/TUWMPI/tuwoptimalgathertrees.c}}.

The completion times computed and shown below are model costs, and not
results from actual executions on any real, parallel computing
system. Later, we do compare actual gather and scatter times for 
specially structured problems using trees with optimal structure
against trees constructed by the problem-aware, adaptive, binomial
tree algorithm~\cite{Traff18:irreggatscat} on a real cluster system.

Concretely, we study the trees constructed by the following
algorithms.
\begin{itemize}
\item
Linear: Linear, star-tree gather algorithm; this construction takes
$O(p)$ steps to compute for an externally given, fixed root $r$.
\item
Binary: Optimal binary tree computed by dynamic programming as
stated in Theorem~\ref{thm:binaryordered}. 
\item
Oblivious: Standard, problem-oblivious, rank-ordered binomial tree.
This can be computed in $O(p)$ steps for an externally given, fixed
root $r$, see, \eg,~\cite{ChanHeimlichPurkayasthavandeGeijn07}.
\item
Adaptive: Problem-aware, adaptive binomial tree constructed by the
algorithm of~\cite{Traff18:irreggatscat} with the properties described
in Proposition~\ref{prop:linear}. Since the constructed tree is
binomial, the root has $\ceiling{\log_2 p}$ children, but otherwise
there are no waiting times and the gather completion time is
$\ceiling{\log_2 p}\alpha+\beta\sum_{i\neq r}m_i+\gamma m_r$ when the
root $r$ is chosen by the algorithm. Also this construction takes
$O(p)$ steps for an externally given, fixed root $r$.
\item
Optimal: Optimal, variable degree gather tree computed by the dynamic
programming algorithm as stated in Theorem~\ref{thm:optimalhomo} in
$O(p^3)$ steps.
\end{itemize}
All algorithms have two variants depending on whether a fixed gather
root $r$ is externally given (as in the \texttt{MPI\_Gatherv}
operation), or a best possible root chosen by the algorithm. We run
both variants in our experiments; it is interesting to see how
negatively a fixed root affects the gather times.

We have experimented with the same distributions of data block sizes
$m_i$ to the processors as in~\cite{Traff18:irreggatscat}. Let $b,
b>0$ be a chosen, average block size, and $\rho>0$ a further
parameter. The problems to be solved have the following distributions
of data block sizes.
\begin{itemize}
\item
\textbf{Same}: For processor $i$, $m_i=b$.
\item
\textbf{Random}: Each $m_i$ is chosen uniformly at random in the range
$[1,2b]$.
\item
\textbf{Random, decreasing}: As random, but the $m_i$ are sorted decreasingly.
\item
\textbf{Random, increasing}: As random, but the $m_i$ are sorted increasingly.
\item
\textbf{Random bucket}: Each $m_i$ has a fixed contribution of
$\ceiling{b/2}$ plus a variable contribution chosen randomly in the
range $[1,b]$.
\item
\textbf{Spikes}: Each $m_i$ is either $\rho b$ or $1$, chosen
  randomly with probability $1/\rho$ for each processor $i$.
\item
\textbf{Decreasing}: For processor $i$, $m_i=\floor{\frac{2b(p-i)}{p}}+1$
\item
\textbf{Increasing}: For processor $i$, $m_i=\floor{\frac{2b(i+1)}{p}}+1$
\item
\textbf{Alternating}: For even numbered processors, $m_i=b+\floor{b/2}$,
  for odd numbered processors $m_i=b-\floor{b/2}$.
\item
\textbf{Skewed}: For the first $\rho$ processors, $m_i=pb/\rho,
0\leq i<\rho$, for the remainder processors $m_i=1$.
\item
\textbf{Two blocks}: All $m_i=0$, except $m_0=\floor{pb/2}$ and
$m_{p-1}=\floor{pb/2}$.
\end{itemize}

The problems are so defined that the total problem size
$m=\sum_{i=0}^{p-1}m_i$ is roughly the same, $m\approx pb$ for all
problem types.  The \textbf{Same} sized problem can serve as a sanity
check, since the optimal completion times are known analytically for
the regular gather operations (binomial trees). For the other problem
distributions we look at variants where the data block sizes are not
sorted and variants where the data blocks are in either increasing or
decreasing order over the processors. This tests the heuristic
proposed in~\cite{Traff04:gatscat} which suggests to virtually rerank
processors such that blocks are in decreasing order and construct
trees over the virtual ranks. Such trees are non-ordered in the sense
of Definition~\ref{def:ordering}. If there are differences between
increasingly and decreasingly sorted problems, this gives concrete
problem instances where non-ordered trees have lower completion times
than ordered ones.

We have computed the completion times of the trees in the
linear transmission cost model with $p=2000$ processors and different
values for $\alpha, \beta$ and $\gamma$, that is for fully connected
systems with homogeneous transmission costs. In all cases, the
externally given root $r$ is chosen as the processor in the middle,
namely $r=\floor{p/2}$ (as was also done
in~\cite{Traff18:irreggatscat}). The parameter for the \textbf{Spikes}
and \textbf{Skewed} distributions has been taken as $\rho=5$.

In the experiments, we have fixed the communication time per unit to
$\beta=1$, and look at results for a low-latency regime with
$\alpha=100$, which means that (only) $100$ units need to be
transferred in order to outweigh the cost of one extra communication
operation. We give results for medium large problems with $b=1000$. We
report the case with local copy costs with $\gamma=\beta$ and the case
with no local copy costs with $\gamma=0$.  The results are shown in
Table~\ref{tab:res-lowlatency1000}; results for ultra low and high
latency networks with $\alpha=1$ and $\alpha=1000$ are shown in
Table~\ref{tab:res-ultralatency1000} and~\ref{tab:res-highlatency1000}
in the appendix.  The first line for each data block distribution is
the completion times for the trees with externally given root
$r=\floor{p/2}$, while the second line gives the best possible time
with a root chosen by the algorithm (given in brackets).
Proposition~\ref{prop:linear} gives another sanity check for the
adaptive binomial tree construction, where the completion times should
be exactly $m+\ceiling{\log_2 p}\alpha=m+1100$.

For all algorithms there is a sometimes quite significant difference
between the case where the algorithm is allowed to choose a best
possible root (leading to lowest completion time) and the case where
the root is externally imposed. The optimal algorithm is best able to
compensate for this by finding trees that better accommodate an
ill-chosen root, and the differences between the two cases for this
algorithm seem small. Since this is the common case for interfaces
like MPI, optimal, ordered, variable-degree trees can have some
advantage over the other constructions.  An externally determined root
is particularly damaging to the oblivious algorithms like the
standard, binomial tree, where this can lead to large data blocks
having to be sent along paths of logarithmic length.  The second best
algorithm is the adaptive binomial tree which can in many cases also
accommodate, and is almost always strictly better than the oblivious
binomial tree.  In the presence of latency, the linear algorithm is
not optimal since it pays at least $(p-1)\alpha$ time units of latency
which becomes more and more prominent with increasing network
latency. Binary trees are in (almost) all cases significantly worse
than the other trees (including linear, star-shaped), and should be
disregarded for gather and scatter collectives.

For the \textbf{Same} sized problems, binomial trees are known to be
optimal, and as can be seen the dynamic programming algorithms
produces trees with the same completion times. For most of the other
distributions, except \textbf{Alternating} and \textbf{Two blocks},
optimal, ordered trees can actually save a few latency terms over the
adaptive, binomial trees, which could be significant in high latency
systems.  Differences between the oblivious and the adaptive binomial
trees can be very considerable when the root $r$ is imposed from the
outside, exhibited prominently for the \textbf{Skewed} and the
\textbf{Two blocks} distributions. The differences mostly disappear if
the algorithms themselves are allowed to choose a best possible
root. Whether block sizes occur in decreasing or increasing order also
makes a considerable difference for the binomial tree algorithms when
the root is fixed as can be seen for the \textbf{Random},
\textbf{Random increasing} and \textbf{Random decreasing}, and for the
\textbf{Increasing} and \textbf{Decreasing} distributions.  This
vindicates the heuristic suggested in~\cite{Traff04:gatscat} according
to which processors should be ordered such that block sizes are in
decreasing order from the root processor. When the root is instead chosen by
the algorithm, the differences disappear for the optimal, ordered
trees and the adaptive trees, but not for the oblivious binomial
trees. It is also noteworthy that sorting, whether in increasing or
decreasing order make the optimal ordered trees better than when the
blocks are not sorted, even if it is only by one $\alpha$.

Similar behavior is seen when there is no local copy cost, $\gamma=0$,
but the costs of the trees are lower by at most the size of one $m_i$
(smallest) block which for the \textbf{Skewed} distributions can be
quite significant. The result for the \textbf{Bucket} distribution
shows that the trees with and without local copy cost do not
necessarily have the same structure. A different, best root is chosen
by the optimal algorithms for the two cases. It is noteworthy, 
that for the decreasing and increasing distributions, the optimal
trees have different cost, even for the case where the algorithm
chooses the root. Again, this vindicates the heuristic that builds
trees over processes in decreasing block order~\cite{Traff04:gatscat},
and shows that dropping the ordering constraint can lead to less
costly trees.

\subsection{Practical impact}

\begin{figure*}[ht]
\includegraphics[width=0.45\textwidth]{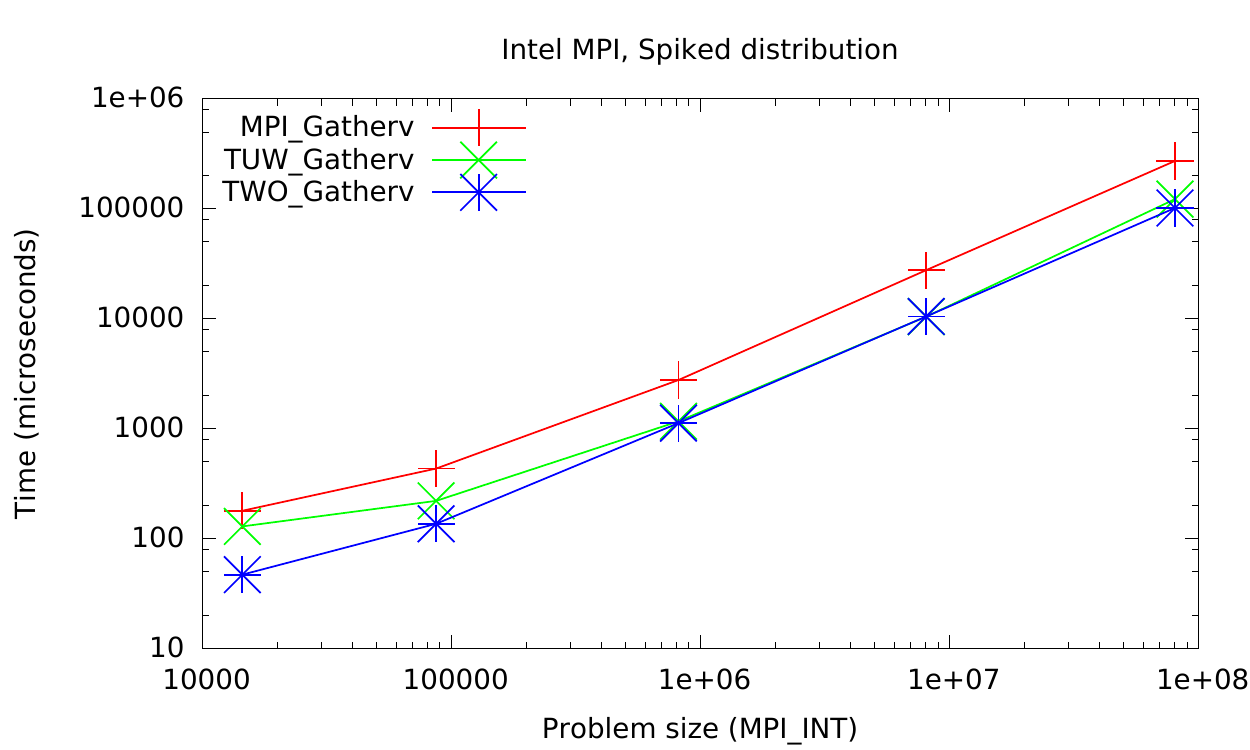}
\includegraphics[width=0.45\textwidth]{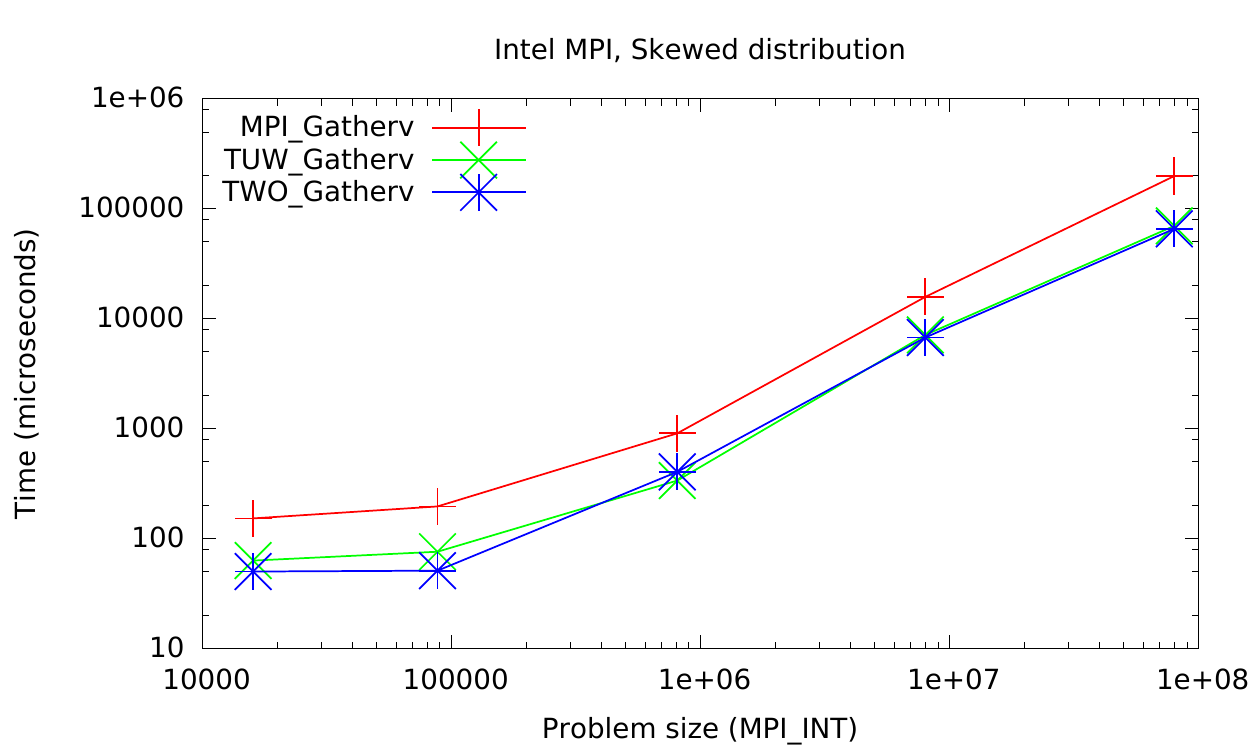}
\caption{Running times of three different irregular gather
  implementations \mpigatherv, \tuwgatherv, and \twogatherv with
  $p=500\times 16=8000$ processes on an Intel/InfiniBand cluster with
  the Intel MPI library.  Minimum completion times over 75 repetitions
  are in microseconds ($\mu s$), problem sizes in number of \mpiint
  elements, doubly logarithmic plot.}
\label{fig:VSCgather}
\end{figure*}

\begin{figure*}[ht]
\includegraphics[width=0.45\textwidth]{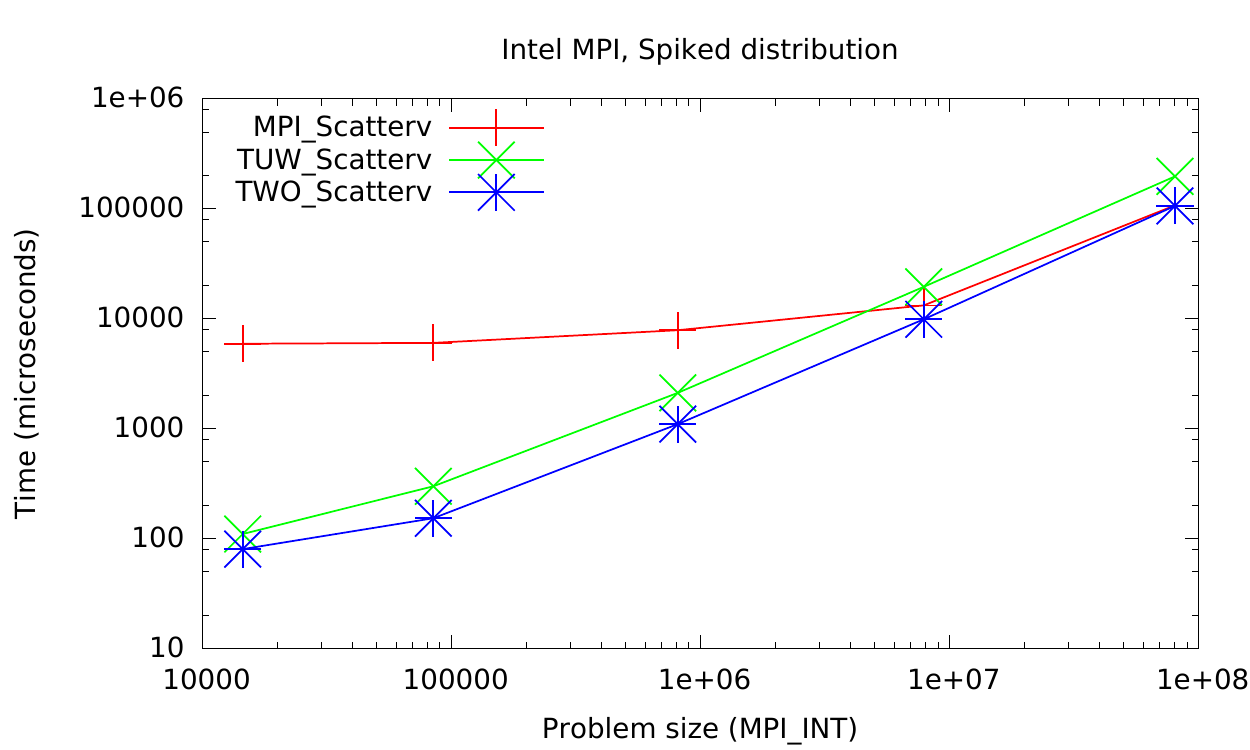}
\includegraphics[width=0.45\textwidth]{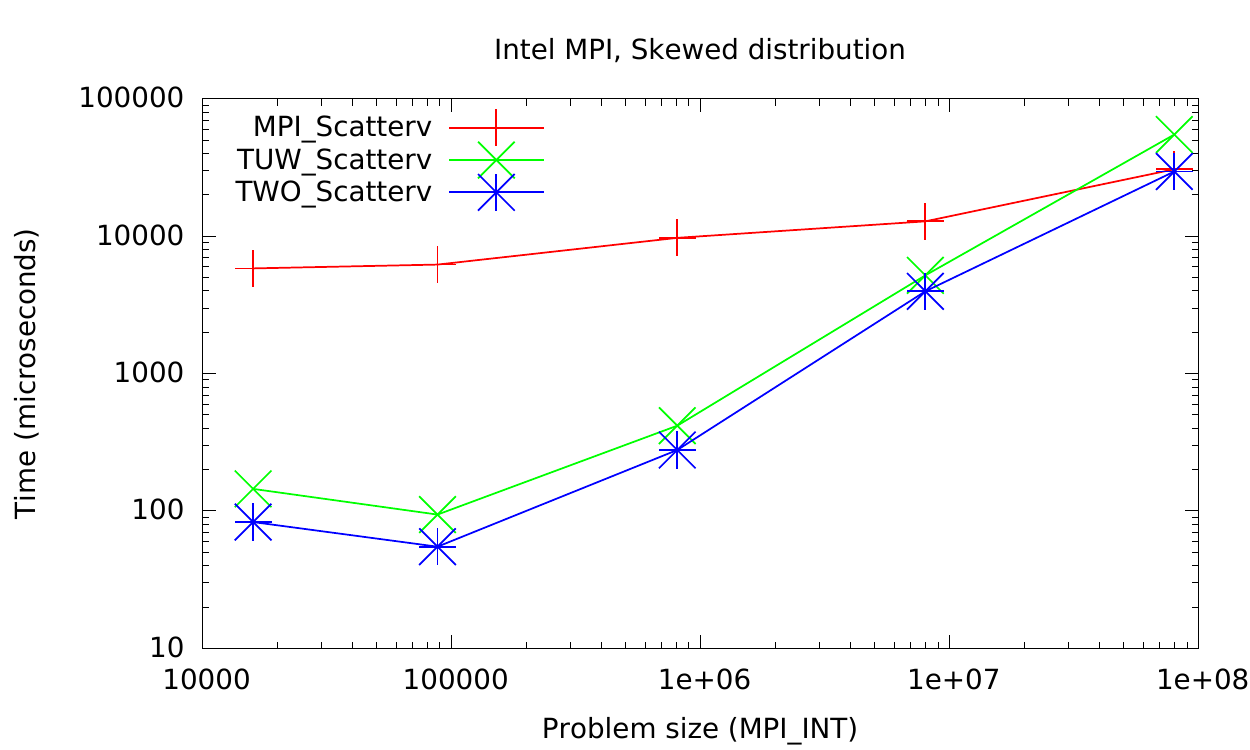}
\caption{Running times of three different irregular scatter
  implementations \mpiscatterv, \tuwscatterv, and \twoscatterv with
  $p=500\times 16=8000$ processes on an Intel/InfiniBand cluster with
  the Intel MPI library.  Minimum completion times over 75 repetitions
  are in microseconds ($\mu s$), problem sizes in number of \mpiint
  elements, doubly logarithmic plot.}
\label{fig:VSCscatter}
\end{figure*}

We finally investigate whether there are cases where non-binomial,
optimal trees can do better than the adaptive binomial trees
described and evaluated in~\cite{Traff18:irreggatscat} when
implemented for a real system for and in MPI. We look at specific
gather/scatter problems where it is obvious that the smallest
completion time trees are not binomial. This is captured in
Lemma~\ref{lem:largesmall}.

\begin{lemma}
\label{lem:largesmall}
Let $x$ and $y$ be two distinct block sizes, $x<y$, for a gather or
scatter problem with block sizes $m_i, 0\leq i<p$ with $m_r=y$ for
exactly one $r$ and $m_i=x$ for all other $i\neq r$, such that the
size of the problem is $m=(p-1)x+y$. Let $\alpha,\beta,\gamma$ be the
communication parameters of a homogeneous system with $\gamma=\beta$.
If $\ceiling{\log (p-1)}\alpha+\beta (p-1)x \leq \beta y$, the cost
of an optimal gather or scatter tree is $\alpha+\beta m$, and smaller
than $\ceiling{\log p}\alpha+\beta m$.
\end{lemma}
\begin{proof}
By Proposition~\ref{prop:linear}, a tree of cost at most $\ceiling{\log
  (p-1)}\alpha+\beta (p-1)x$ exists over the processors with
$m_i=x$. Concurrently, the processor with $m_r=y$ can perform its local
copy which takes $\beta y$, after which transmission between the roots
of the two trees take place. Since there is only one communication
latency to be accounted for in this tree, no better, lower latency
tree can exist since also $p>1$. Any tree where the block $m_r$ is
transmitted will have cost $\alpha+\beta m_r$ plus the cost of the
other blocks, and thus larger.
\end{proof}

The argument assumes that processors can be freely ordered such that
the $x$ and $y$ trees can be handled concurrently. If ordered gather
or scatter trees are required, this is only possible if either $m_0=y$
or $m_{p-1}=y$, otherwise at least two communication operations will
be required. This again shows that ordered trees can be slower than
non-ordered trees.  The dynamic programming algorithms would produce
the best possible, ordered trees for such problem instances.

For problems with only two block sizes, optimal trees as outlined in
Lemma~\ref{lem:largesmall} can readily be implemented by two
concurrent gather or scatter operations on the domains of small and
large blocks, respectively, and a single point-to-point
communication. Since both gather/scatter operations are regular (small
or large blocks), we can use the best implementations available for
these problem, \eg, the library native \mpigather and \mpiscatter
operations.  We use these implementations which we call \twogatherv
and \twoscatterv for the \textbf{Spikes} and \textbf{Skewed}
distributions to compare against MPI library native \mpigatherv and
\mpiscatterv operations, and against implementations of the adaptive
binomial tree algorithm called \tuwgatherv and
\tuwscatterv~\cite{Traff18:irreggatscat}. We have done experiments on
a medium-large Intel/InfiniBand cluster with 2000 Dual Intel Xeon
E5-2650v2 8-core processors running at 2.6GHz, interconnected with an
InfiniBand QDR-80 network\footnote{This is the so-called Vienna
  Scientific Cluster, the VSC3, see \url{vsc.ac.at}.  The author
  thanks for support and access to this machine. The code used for
  these experiments is available under
  \url{par.tuwien.ac.at/Downloads/TUWMPI/tuwgatherv.c}.}.  The MPI
library is the native Intel MPI version 2018, and we choose the
binomial tree algorithm for \mpigatherv and the linear algorithm for
\mpiscatterv which seemed to be the best performing available
implementations in this library. We used $p=8000$ MPI processes on 500
nodes with 16 processes running on each.

Results for $b=1,10,100,1000,10000$ are plotted in
Figure~\ref{fig:VSCgather} and Figure~\ref{fig:VSCscatter}, and show
that optimal trees can perform better than binomial trees by a
significant percentage, in the experiments with the \textbf{Spikes}
scatter problems ranging from 25\% to 50\%.  The poor performance of
\mpiscatterv on this system is due to a linear algorithm, which is
clearly not the right choice for small(er), irregular scatter
problems.

\section{Concluding remarks}

This paper investigated the irregular scatter and gather collective
communication operations, both with respect to completion times for
the operations, as well as with respect to the difficulty of finding
communication trees leading to good completion times. The results
under a synchronous, one-ported point-to-point communication model
show that there is a difference in both respects between ordered and
non-ordered communication trees as introduced here: Strongly ordered,
minimal completion time trees can be computed in polynomial time,
whereas constructing possibly better, non-ordered trees is NP-hard. We
implemented dynamic programming algorithms for computing optimal
ordered trees, and compared the quality (completion times) of the
constructed trees for a set of different problem data block
distributions. Based on these experiments, the problem-dependent,
adaptive binomial tree construction which can compute trees fast in a
distributed manner~\cite{Traff18:irreggatscat} can be seen to produce
ordered trees that are very close to the optimal completion time trees
and probably sufficient for all practical purposes. However,
experiments with a real implementation for specially structured
problems show that there are practical cases where optimal algorithms
can do significantly better. This leaves room for devising practical
algorithms that come still closer to the optimum solutions.

A number of interesting open problems has emerged that it would be
worthwhile and fruitful to pursue further.
\begin{itemize}
\item
  Can the dynamic programming constructions be improved and extended?
  Monotonicity properties in the gather and scatter trees can possibly
  lower the complexity. Which alternative approaches can possibly lead
  to optimal or approximately optimal trees?
\item
  What is the complexity of constructing weakly ordered communication
  trees, even in the one-ported model? Weakly ordered gather and
  scatter trees store consecutive segments of data blocks at all
  processors, but allow segments to be communicated in any order, not
  strictly in increasing rank order as required by the strictly
  ordered trees. Resolving this question could throw light on the
  problems in more asynchronous communication models.
\item
  Under what conditions are trees not optimal communication structures
  for the gather and scatter problems? When might directed acyclic
  graphs (DAGs) perform better? How can such DAGs be constructed?
\item
  How difficult are the problems with $k>1$ communication ports? How
  difficult are the problems under asynchronous communication models
  permitting overlap of communication operations?
\end{itemize}

\bibliographystyle{plain}
\bibliography{traff,parallel}

\appendix
\label{app:moreresults}

\begin{table*}[ht]
\caption{Gather tree completion times for ultra low latency network,
  $p=2000,b=1000,r=\floor{p/2},\rho=5,\alpha=1,\beta=1,\gamma=1,0$. For
  the different tree constructions, both the cases where the root $r$
  is given externally (first row), and the case where a best root is
  chosen by the algorithm (second row) are shown. For the latter, the
  chosen root $(r)$ is given in brackets.}
\label{tab:res-ultralatency1000}
\begin{center}
\begin{tiny}
  \begin{tabular}{lrrrrrrrrrrr}
  \hline
  $\gamma=1$ & $m$ & Linear & ($r$) & Binary & ($r$) & Oblivious & ($r$) & Adaptive & ($r$) & Optimal & ($r$) \\
   Same & 2000000 & 2001999 & & 3610016 & & 2000011 & & 2000011 & & 2000011 & \\
      & 2000000 & 2001999 & (0) & 3226015 & (0) & 2000011 & (0) & 2000011 & (1023) & 2000011 & (1998) \\
 Random & 1983668 & 1985667 & & 3586046 & & 2030136 & & 1993762 & & 1983680 & \\
      & 1983668 & 1985667 & (0) & 3196585 & (753) & 2020053 & (511) & 1983679 & (331) & 1983679 & (1171) \\
 Random & 1983668 & 1985667 & & 4374656 & & 2338418 & & 2248476 & & 1983680 & \\
 decreasing     & 1983668 & 1985667 & (0) & 3192367 & (0) & 1983679 & (0) & 1983679 & (1) & 1983678 & (0) \\
 Random & 1983668 & 1985667 & & 3881291 & & 3022129 & & 2929259 & & 1983680 & \\
 increasing     & 1983668 & 1985667 & (0) & 3199725 & (1236) & 2076549 & (1999) & 1983679 & (1791) & 1983678 & (1998) \\
 Bucket & 2005668 & 2007667 & & 3612247 & & 2032024 & & 2012762 & & 2005680 & \\
      & 2005668 & 2007667 & (0) & 3233050 & (0) & 2032024 & (0) & 2005679 & (215) & 2005679 & (682) \\
 Spikes & 2001600 & 2003599 & & 3567899 & & 2141583 & & 2036604 & & 2001613 & \\
      & 2001600 & 2003599 & (0) & 3217645 & (775) & 2131633 & (1999) & 2001611 & (1549) & 2001610 & (1198) \\
  Decreasing & 2003000 & 2004999 & & 4415103 & & 2352535 & & 2265155 & & 2003012 & \\
      & 2003000 & 2004999 & (0) & 3223175 & (0) & 2003011 & (0) & 2003011 & (1) & 2003010 & (0) \\
  Increasing & 2003000 & 2004999 & & 3915600 & & 3047491 & & 2954363 & & 2003012 & \\
      & 2003000 & 2004999 & (0) & 3232091 & (1999) & 2096139 & (1999) & 2003011 & (1791) & 2003010 & (1998) \\
  Alternating & 2000000 & 2001999 & & 3609016 & & 2000011 & & 2000011 & & 2000011 & \\
      & 2000000 & 2001999 & (0) & 3226015 & (0) & 2000011 & (0) & 2000011 & (1023) & 2000011 & (1998) \\
Skewed & 2001995 & 2003994 & & 4402995 & & 17201968 & & 4002001 & & 2002010 & \\
      & 2001995 & 2003994 & (0) & 2401998 & (0) & 2002006 & (0) & 2002006 & (3) & 2001998 & (2) \\
  Two blocks & 2000000 & 2000002 & & 3000002 & & 11000011 & & 3000002 & & 2000002 & \\
      & 2000000 & 2000001 & (0) & 2000001 & (0) & 2000001 & (0) & 2000001 & (1999) & 2000001 & (0) \\
\hline
$\gamma=0$ & $m$ & Linear & ($r$) & Binary & ($r$) & Oblivious & ($r$) & Adaptive & ($r$) & Optimal & ($r$) \\
 Same & 2000000 & 2000999 & & 3608016 & & 1999011 & & 1999011 & & 1999011 & \\
      & 2000000 & 2000999 & (0) & 3225015 & (0) & 1999011 & (0) & 1999011 & (1023) & 1999011 & (1998) \\
 Random & 1983668 & 1984973 & & 3578075 & & 2028324 & & 1991822 & & 1982985 & \\
      & 1983668 & 1983997 & (1996) & 3190317 & (0) & 2018241 & (511) & 1981739 & (330) & 1981679 & (891) \\
 Random & 1983668 & 1984688 & & 4371920 & & 2336421 & & 2246476 & & 1982701 & \\
 decreasing     & 1983668 & 1983667 & (0) & 3190627 & (0) & 1981679 & (0) & 1981679 & (0) & 1981678 & (0) \\
 Random & 1983668 & 1984688 & & 3876934 & & 3020146 & & 2927474 & & 1982701 & \\
 increasing     & 1983668 & 1983667 & (1999) & 3194265 & (1999) & 2074566 & (1999) & 1981894 & (1791) & 1981681 & (1998) \\
 Bucket & 2005668 & 2006474 & & 3607205 & & 2030782 & & 2011266 & & 2004487 & \\
      & 2005668 & 2006334 & (1999) & 3229693 & (1999) & 2030782 & (0) & 2004183 & (215) & 2004180 & (891) \\
 Spikes & 2001600 & 2003598 & & 3542954 & & 2136583 & & 2031604 & & 2001610 & \\
      & 2001600 & 1998599 & (1992) & 3192623 & (775) & 2126633 & (1999) & 1996611 & (1549) & 1996610 & (1198) \\
  Decreasing & 2003000 & 2003998 & & 4412223 & & 2350535 & & 2263154 & & 2002011 & \\
      & 2003000 & 2002998 & (0) & 3221689 & (0) & 2001010 & (0) & 2001010 & (0) & 2001009 & (0) \\
  Increasing & 2003000 & 2003997 & & 3911230 & & 3045506 & & 2952570 & & 2002010 & \\
      & 2003000 & 2002998 & (1999) & 3226539 & (1999) & 2094154 & (1999) & 2001218 & (1791) & 2001010 & (1998) \\
  Alternating & 2000000 & 2000499 & & 3603016 & & 1999511 & & 1998511 & & 1998511 & \\
      & 2000000 & 2000499 & (0) & 3221515 & (0) & 1999511 & (0) & 1998511 & (1022) & 1998511 & (1998) \\
Skewed & 2001995 & 2003993 & & 4002994 & & 16801968 & & 3602001 & & 2002007 & \\
      & 2001995 & 1603994 & (0) & 2001998 & (0) & 1602006 & (0) & 1602006 & (3) & 1601998 & (2) \\
  Two blocks & 2000000 & 2000002 & & 2000002 & & 11000011 & & 2000002 & & 2000002 & \\
      & 2000000 & 1000001 & (0) & 1000001 & (0) & 1000001 & (0) & 1000001 & (1999) & 1000001 & (0) \\
\hline
  \end{tabular}
\end{tiny}
\end{center}
\end{table*}

\begin{table*}[ht]
\caption{Gather tree completion times for high latency network,
  $p=2000,b=1000,r=\floor{p/2},\rho=5,\alpha=1000,\beta=1,\gamma=1,0$. For
  the different tree constructions, both the cases where the root $r$
  is given externally (first row), and the case where a best root is
  chosen by the algorithm (second row) are shown. For the latter, the
  chosen root $(r)$ is given in brackets.}
\label{tab:res-highlatency1000}
\begin{center}
\begin{tiny}
  \begin{tabular}{lrrrrrrrrrrr}
  \hline
  $\gamma=1$ & $m$ & Linear & ($r$) & Binary & ($r$) & Oblivious & ($r$) & Adaptive & ($r$) & Optimal & ($r$) \\
   Same & 2000000 & 3999000 & & 3626000 & & 2011000 & & 2011000 & & 2011000 & \\
      & 2000000 & 3999000 & (0) & 3241000 & (0) & 2011000 & (0) & 2011000 & (1023) & 2011000 & (1998) \\
 Random & 1983668 & 3982668 & & 3601460 & & 2041125 & & 2004751 & & 1995061 & \\
      & 1983668 & 3982668 & (0) & 3211992 & (753) & 2031042 & (511) & 1994668 & (331) & 1994668 & (1014) \\
 Random & 1983668 & 3982668 & & 4389641 & & 2349407 & & 2259465 & & 1994878 & \\
 decreasing     & 1983668 & 3982668 & (0) & 3206353 & (0) & 1994668 & (0) & 1994668 & (1) & 1993884 & (0) \\
 Random & 1983668 & 3982668 & & 3896276 & & 3033118 & & 2940248 & & 1994878 & \\
 increasing     & 1983668 & 3982668 & (0) & 3216898 & (1236) & 2087538 & (1999) & 1994668 & (1791) & 1993884 & (1998) \\
 Bucket & 2005668 & 4004668 & & 3627966 & & 2043013 & & 2023751 & & 2017333 & \\
      & 2005668 & 4004668 & (0) & 3248035 & (0) & 2043013 & (0) & 2016668 & (215) & 2016668 & (1006) \\
 Spikes & 2001600 & 4000600 & & 3583827 & & 2152572 & & 2047593 & & 2014600 & \\
      & 2001600 & 4000600 & (0) & 3232618 & (775) & 2142622 & (1999) & 2012600 & (1549) & 2011600 & (1198) \\
  Decreasing & 2003000 & 4002000 & & 4430088 & & 2363524 & & 2276144 & & 2014256 & \\
      & 2003000 & 4002000 & (0) & 3237161 & (0) & 2014000 & (0) & 2014000 & (1) & 2013649 & (0) \\
  Increasing & 2003000 & 4002000 & & 3930239 & & 3058480 & & 2965352 & & 2014256 & \\
      & 2003000 & 4002000 & (0) & 3247473 & (1233) & 2107128 & (1999) & 2014000 & (1791) & 2013649 & (1998) \\
  Alternating & 2000000 & 3999000 & & 3625000 & & 2011000 & & 2011000 & & 2011000 & \\
      & 2000000 & 3999000 & (0) & 3241000 & (0) & 2011000 & (0) & 2011000 & (1023) & 2011000 & (1998) \\
Skewed & 2001995 & 4000995 & & 4407990 & & 17212957 & & 4012990 & & 2016995 & \\
      & 2001995 & 4000995 & (0) & 2404995 & (0) & 2012995 & (0) & 2012995 & (3) & 2004995 & (2) \\
  Two blocks & 2000000 & 2002000 & & 3002000 & & 11011000 & & 3002000 & & 2002000 & \\
      & 2000000 & 2001000 & (0) & 2001000 & (0) & 2001000 & (0) & 2001000 & (1999) & 2001000 & (0) \\
\hline
$\gamma=0$ & $m$ & Linear & ($r$) & Binary & ($r$) & Oblivious & ($r$) & Adaptive & ($r$) & Optimal & ($r$) \\
 Same & 2000000 & 3998000 & & 3624000 & & 2010000 & & 2010000 & & 2010000 & \\
      & 2000000 & 3998000 & (0) & 3240000 & (0) & 2010000 & (0) & 2010000 & (1023) & 2010000 & (1998) \\
 Random & 1983668 & 3981974 & & 3594007 & & 2039313 & & 2002811 & & 1994014 & \\
      & 1983668 & 3980998 & (1996) & 3205203 & (0) & 2029230 & (511) & 1992728 & (330) & 1992668 & (891) \\
 Random & 1983668 & 3981689 & & 4386905 & & 2347410 & & 2257465 & & 1993877 & \\
 decreasing     & 1983668 & 3980668 & (0) & 3204613 & (0) & 1992668 & (0) & 1992668 & (0) & 1992154 & (0) \\
 Random & 1983668 & 3981689 & & 3891919 & & 3031135 & & 2938463 & & 1993877 & \\
 increasing     & 1983668 & 3980668 & (1999) & 3211188 & (1999) & 2085555 & (1999) & 1992883 & (1791) & 1992157 & (1998) \\
 Bucket & 2005668 & 4003475 & & 3622985 & & 2041771 & & 2022255 & & 2015857 & \\
      & 2005668 & 4003335 & (1999) & 3244678 & (1999) & 2041771 & (0) & 2015172 & (215) & 2015169 & (891) \\
 Spikes & 2001600 & 4000599 & & 3559937 & & 2147572 & & 2042593 & & 2012599 & \\
      & 2001600 & 3995600 & (1992) & 3208607 & (775) & 2137622 & (1999) & 2007600 & (1549) & 2006600 & (1198) \\
  Decreasing & 2003000 & 4000999 & & 4427208 & & 2361524 & & 2274143 & & 2013179 & \\
      & 2003000 & 3999999 & (0) & 3235675 & (0) & 2011999 & (0) & 2011999 & (0) & 2011712 & (0) \\
  Increasing & 2003000 & 4000998 & & 3925677 & & 3056495 & & 2963559 & & 2013179 & \\
      & 2003000 & 3999999 & (1999) & 3241774 & (1999) & 2105143 & (1999) & 2012207 & (1791) & 2011713 & (1998) \\
  Alternating & 2000000 & 3997500 & & 3619000 & & 2010500 & & 2009500 & & 2009500 & \\
      & 2000000 & 3997500 & (0) & 3236500 & (0) & 2010500 & (0) & 2009500 & (1022) & 2009500 & (1998) \\
Skewed & 2001995 & 4000994 & & 4007989 & & 16812957 & & 3612990 & & 2014994 & \\
      & 2001995 & 3600995 & (0) & 2004995 & (0) & 1612995 & (0) & 1612995 & (3) & 1604995 & (2) \\
  Two blocks & 2000000 & 2002000 & & 2002000 & & 11011000 & & 2002000 & & 2002000 & \\
      & 2000000 & 1001000 & (0) & 1001000 & (0) & 1001000 & (0) & 1001000 & (1999) & 1001000 & (0) \\
\hline
\end{tabular}
\end{tiny}
\end{center}
\end{table*}

\end{document}